\documentclass[twocolumn,10pt]{IEEEtran}
\usepackage{braket}

\input{def.tex}
\usepackage{dsfont}
\usepackage{minibox}
\DeclareSymbolFont{matha}{OML}{txmi}{m}{it}% txfonts
\DeclareMathSymbol{\varv}{\mathord}{matha}{118}
\usepackage{eurosym}
\usepackage{hhline,physics}

\usepackage{multicol}
\usepackage{algorithm}
\usepackage{graphicx}
\usepackage{textcomp}
\usepackage{caption,pifont}
\usepackage[multiple]{footmisc}

\usepackage{algorithmic}

\newcommand{\trans}{^{\mathsf{T}}}
\newcommand{\herm}{^{\H}}
\makeatletter

\IEEEoverridecommandlockouts
\begin{document}
	\title{\huge{Achievable Rate of a STAR-RIS Assisted Massive MIMO System Under Spatially-Correlated Channels}} 
	\author{Anastasios Papazafeiropoulos, Le-Nam Tran, Zaid Abdullah, Pandelis Kourtessis, Symeon Chatzinotas \thanks{A. Papazafeiropoulos is with the Communications and Intelligent Systems Research Group, University of Hertfordshire, Hatfield AL10 9AB, U. K., and with SnT at the University of Luxembourg, Luxembourg. L. N. Tran is with the School of Electrical and Electronic Engineering, University College Dublin, Ireland. P. Kourtessis is with the Communications and Intelligent Systems Research Group, University of Hertfordshire, Hatfield AL10 9AB, U. K. S. Chatzinotas is with the SnT at the University of Luxembourg, Luxembourg. A. Papazafeiropoulos was supported  by the University of Hertfordshire's 5-year Vice Chancellor's Research Fellowship.
			S. Chatzinotas   was supported by the National Research Fund, Luxembourg, under the project RISOTTI. E-mails: tapapazaf@gmail.com, nam.tran@ucd.ie, p.kourtessis@herts.ac.uk, \{zaid.abdullah,symeon.chatzinotas\}@uni.lu.}}
	\maketitle\vspace{-1.7cm}
	\begin{abstract}
		Reconfigurable intelligent surfaces (RIS)-assisted  massive multiple-input multiple-output (mMIMO)  is a promising technology for applications in next-generation networks. However, reflecting-only RIS provides limited coverage compared to a simultaneously transmitting and reflecting RIS (STAR-RIS). Hence, in this paper, we focus on  the downlink achievable rate and its optimization of a STAR-RIS-assisted mMIMO  system. Contrary to previous works on STAR-RIS, we consider mMIMO, correlated fading, and multiple user equipments (UEs) at both sides of the RIS. In particular, we introduce an estimation approach of the aggregated channel with the main benefit of reduced overhead  links instead of estimating the individual  channels.  	{Next, leveraging channel hardening in mMIMO and the use-and-forget bounding technique, we obtain an achievable rate in closed-form that only depends on statistical channel state information (CSI). To optimize the amplitudes and phase shifts of the STAR-RIS, we employ a projected gradient ascent method (PGAM) that simultaneously adjusts the amplitudes and phase shifts for both energy splitting (ES) and mode switching (MS) STAR-RIS operation protocols.} By considering large-scale fading, the proposed optimization can be performed every several coherence intervals, which can significantly reduce overhead. Considering that STAR-RIS has  twice the number of controllable parameters 
		compared to conventional reflecting-only RIS, this accomplishment offers substantial practical benefits. Simulations are carried out to verify the analytical results, reveal the interplay of the achievable rate with fundamental parameters, and show the superiority of STAR-RIS regarding its achievable rate compared to its  reflecting-only counterpart.
	\end{abstract}\vspace{-20pt}
	
	\section{Introduction}
	Reconfigurable intelligent surfaces (RIS) have emerged as a promising technology to meet the requirements of  sixth-generation (6G) networks such as a 1000-fold capacity increase together with increased connectivity among billions of devices \cite{Wu2019,Basar2019,Kisseleff2020}. A RIS consists of a metamaterial layer of low-cost controllable elements. Among its significant benefits is that its control signals can be dynamically adjusted to steer the impinging waves in specific directions and shape the propagation environment while providing uninterrupted service not only with low hardware cost, but also with low power consumption due to the absence of any power amplifiers.
	
	Most of the existing works on RIS have assumed that both the transmitter and the receiver are found on the same side of the surface, i.e., only reflection takes place \cite{Wu2019,Basar2019,Bjoernson2020,VanChien2021,Papazafeiropoulos2021,Papazafeiropoulos2022}. However, practical applications might include user equipments (UEs) on both sides of the RIS, which contain the spaces in front and behind the surface. Recently, advancements in programmable metamaterials have enabled the technology of simultaneously transmitting and reflecting RIS (STAR-RIS).\footnote{We note that the word ``transmitting'' does not correspond to active transmission but implies coverage of the UEs   at the other side of the RIS.} Hence, STAR-RIS has been proposed as a technology to satisfy this demand, i.e., it provides full space coverage by changing the amplitudes and phases of the impinging waves     \cite{Xu2021,Mu2021,Niu2021,Wu2021,Niu2022}. For instance,  in \cite{Xu2021}, the authors provided a general 	hardware model and two-channel models corresponding to the near-field region and the far-field region of STAR-RIS with only two UEs. Also, they showed  that the  coverage and diversity gain are greater than  reflecting-only/conventional RIS-assisted systems.  Furthermore, in \cite{Mu2021}, three operating protocols for adjusting the transmission and reflection coefficients of the transmitted and reflected signals were suggested, namely, energy splitting (ES), mode switching (MS), and time switching (TS). 
	
	In particular,	 most existing works on  RIS-aided systems have assumed perfect CSI, but this is a highly unrealistic assumption since practical systems have imperfect CSI. The accuracy of the channel state information (CSI) at the transmitter side is crucial to achieving a high beamforming gain of RIS, which includes channels between the transmitter and the UEs \cite{Zheng2022}. However, the acquisition of CSI is challenging because of the following reasons. First, RIS, in general, consists of passive elements to perform the desired reflecting operation, which makes any active transmission or reception infeasible, i.e., it cannot perform any sampling or processing of the pilots  \cite{Wu2019}. For this reason, an alternative method is the estimation of the aggregated transmitter-RIS-receiver channel by sending appropriate pilot symbols \cite{Yang2020b}. Second, RIS are generally large and consist of a large number of elements, and thus, induce  high training overhead for channel estimation (CE), which results in spectral efficiency (SE) reduction \cite{Wang2020}. 
	
	Various CE schemes have been proposed to address this issue \cite{Zheng2019,Mishra2019,He2019,Elbir2020,Nadeem2020}. For example, in  \cite{Mishra2019}, an ON/OFF CE method was proposed, where the estimates of all RIS-assisted channels for a single-user MISO system are obtained one-by-one. Note that in the case of  multi-user systems, this model was extended, assuming  all RIS elements to be active during training, but  the number of sub-phases is required to be at least equal  to the number of RIS elements \cite{Nadeem2020}. Although that method provides better CE as the number of sub-phases increases, the  achievable rate decreases because the data transmission phase takes a smaller fraction of the coherence time due to excessive training overhead. Also, that method computes the estimates of the channels of the individual RIS elements but the covariance of the channel vector from all RIS elements to a specific UE is unknown. Especially, in the case of STAR-RIS, CE becomes more challenging because UEs are located in both transmission and reflection regions, which requires different passive beamforming matrices (PBMs). In \cite{Wu2021}, a CE scheme was presented but did not account for multiple antennas at the BS, multiple UEs, and correlated fading.
	
	In parallel, many early works on conventional RIS assumed independent Rayleigh fading  such as \cite{Wu2019}, but recently, it was shown that RIS correlation should be considered because it is unavoidable in practical systems \cite{Bjoernson2020}. To this end, several works on conventional RIS have taken into account  the impact of RIS correlation \cite{VanChien2021,Papazafeiropoulos2021}, but only \cite{Wang2021b} has considered fading correlation on a STAR-RIS assisted system. Furthermore, except \cite{Mu2021,Niu2021,Niu2022}, all other works have assumed a single-antenna transmitter. Also, all previous studies on STAR-RIS have only  considered a single UE on each side of the STAR-RIS. In this paper, we consider a more general case where multiple UEs are present on each side of the STAR-RIS.

	\textit{Contributions}:  The observations above indicate the topic of this work, which concerns the study and design of a STAR-RIS assisted mMIMO system under the realistic conditions of imperfect CSI and correlated fading. These realistic assumptions and the consideration  of multiple UEs at each side of the STAR-RIS make it extremely difficult for the derivations of the achievable rate and the resulting optimization of the amplitudes and phase shifts of the STAR-RIS. Our main contributions are summarized as follows:
	\begin{itemize}
		\item  Aiming to characterize the potentials of STAR-RIS under realistic assumptions, we include the effect of spatially correlated fading at both the BS and the STAR-RIS.\footnote{In the case of the active beamforming, being MRT in this work, it is designed based on the instantaneous channel, which depends on the correlation of the aggregated channel described by (6). In the case of the passive beamforming, it is designed based on statistical CSI in terms of path loss and correlation. Specifically, the sum-rate expression depends only on these large scale statistics, which vary every several coherence intervals. Hence, passive beamforming can be optimized at every several coherence intervals. Moreover, given that we rely on  the statistical CSI approach, if no correlation is considered, the aggregated correlation will not depend on the phase shifts, which means that the sum rate cannot be optimized with respect to the phase shifts.}  In particular, we consider a massive multiple-input multiple-output (mMIMO) system with  a BS having a large but finite number of antennas. Under this general setup, we derive the downlink achievable  spectral efficiency (SE) of a STAR-RIS-assisted mMIMO system with imperfect CSI and correlated fading in closed form that depends only on large-scale statistics, which has not been known previously. Moreover, we achieve this by a unified analysis of the channel estimation and data transmission phases for UEs located in either the $t$ or $r$ regions, which distinguishes our work from previous research.
		\item  Contrary to \cite{Xu2021,Mu2021,Niu2021,Wu2021,Niu2022} which have assumed a single UE at each side of the surface, we consider multiple UEs at each side of the RIS,  which are served in the same time-frequency response.
		\item  We apply the linear minimum mean square error (LMMSE) method to perform CE, and obtain closed-form expressions with lower overhead than other CE methods suggested for RIS-assisted  systems. Specifically, we demonstrate that LMMSE can be applied without the need for a tailored design for STAR-RIS under conditions of statistical CSI. Note that previous works do not provide analytical expressions and/or do not take into account  the spatial correlation at the RIS \cite{Xu2021,Mu2021,Niu2021,Wu2021,Niu2022}. 
		\item   Our analysis relies on statistical CSI, meaning that our closed-form expressions are dependent only on large-scale fading that changes at every several coherence intervals. Thus, the proposed optimization of the STAR-RIS can take place at every several coherence intervals, which saves significant overhead. On the contrary, previous studies, which are based on instantaneous CSI changing at each coherence interval, might not be feasible in practice due to inherent large overheads.\footnote{In this work, we have followed the two-timescale transmission protocol approach as  	in \cite{Zhao2020},  where a maximisation of the achievable sum rate of a RIS-assisted multi-user multi-input single-output (MU-MISO) system took place. According to this approach, the precoding is designed in terms of instantaneous CSI, while the RIS   phase shifts is optimized by using statistical CSI. Notably, all works, which are based on statistical CSI, have relied on the two-timescale protocol. Examples are the  study of the impact of hardware impairments on the sum rate and the minimum rate in \cite{Papazafeiropoulos2021} and \cite{Papazafeiropoulos2021b}, respectively.}
		\item  We formulate the problem of finding the amplitudes and phase shifts of the STAR-RIS to maximize the achievable sum SE. {Our optimization framework considers multiple users at each side of the STAR-RIS in a unified manner.} Despite its non-convexity, we derive an iterative efficient method based on the projected gradient ascent method in which both amplitudes and phase shifts of the STAR-RIS are updated simultaneously at each iteration.  To the best of our knowledge, we are the first to optimize simultaneously the amplitudes and the phase shifts of the PBM in a STAR-RIS system. This is a significant contribution since other works optimize only the phase shifts or optimize both the amplitudes and the phase shifts in an  alternating optimization manner. Moreover, as large-scale fading is considered, our optimization has very lower overhead in terms of complexity, training,  and feedback compared to other works which rely on instantaneous CSI such as \cite{Mu2021}. Notably, this property is  important for STAR-RIS applications, which have twice the number of optimization variables compared to reflecting-only RIS. {We also remark that the beamforming optimization based on statistical CSI for STAR-RIS has not been investigated previously.}
		\item Simulations and analytical results are provided to shed light on the impact of various parameters and to show the superiority of STAR-RIS over conventional  RIS. For example, we find that the system  performance decreases as the RIS correlation increases.
	\end{itemize}
	
	\textit{Paper Outline}: The remainder of this paper is organized as follows. Section~\ref{System} presents the system model of a STAR-RIS-assisted mMIMO system with correlated Rayleigh fading. Section~\ref{ChannelEstimation} provides the CE. Section~\ref{PerformanceAnalysis} presents the downlink data transmission with the derived downlink sum SE. Section \ref{PSConfig} provides the simultaneous amplitudes and  phase-shifts configuration concerning both  the PBMs for the  transmission and reflection regions. The numerical results are placed in Section~\ref{Numerical}, and Section~\ref{Conclusion} concludes the paper.
	
	\textit{Notation}: Vectors and matrices are denoted by boldface lower and upper case symbols, respectively. The notations $(\cdot)^\T$, $(\cdot)^\H$, and $\tr\!\left( {\cdot} \right)$ describe the transpose, Hermitian transpose, and trace operators, respectively. Moreover, the notations $ \arg\left(\cdot\right) $,  $\EE\left[\cdot\right]$,  and $ \mathrm{Var}(\cdot) $ express the argument function,  the expectation, and variance operators, respectively. The notation  $\diag\left(\bA\right) $ describes a vector with elements equal to the  diagonal elements of $ \bA $, the notation  $\diag\left(\bx\right) $ describes a diagonal  matrix whose elements are $ \bx $, while  $\bb \sim \cC\cN{(\b0,\mathbf{\Sigma})}$ describes a circularly symmetric complex Gaussian vector with zero mean and a  covariance matrix $\mathbf{\Sigma}$.

	\section{System Model}\label{System}
	We consider a STAR-RIS-aided system,  where a  BS with an $ M $-element uniform linear array (ULA) serves simultaneously $ K $ single-antenna UEs that are distributed on both sides of the STAR-RIS, as illustrated in Fig. \ref{Fig1}. Specifically, $ \mathcal{K}_{t}=\{1,\ldots,K_{t} \} $ UEs are located in the transmission region $ (t) $ and $ \mathcal{K}_{r}=\{1,\ldots,K_{r} \} $ UEs are  located in the reflection region $ (r) $, respectively, where $ K_{t}+K_{r}=K $. Also, we denote by $ \mathcal{W} = \{w_{1}, w_{2}, ..., w_{K}\}  $ the set that defines the RIS operation mode for each of the $ K $ UEs. In particular, if the $ k $th UE is located behind the STAR-RIS (i.e., $ k\in   \mathcal{K}_{t}$), then $ w_{k} = t $, while $ w_{k} = r $ when the $ k $th UE is facing the STAR-RIS (i.e., $ k\in   \mathcal{K}_{r}$).  Moreover, we assume  direct links between the BS and UEs.  The RIS consists of a uniform planar array (UPA) composed of $ N_{\mathrm{h}} $ horizontally 
	and $ N_{\mathrm{v}} $ vertically  passive elements, which belong to the set of $\mathcal{N}=\{1, \ldots, N \}$ elements, where $ N= N_{\mathrm{h}} \times N_{\mathrm{v}} $ is the total number of RIS elements.
	\begin{figure}[!h]
		\begin{center}
			\includegraphics[width=0.9\linewidth]{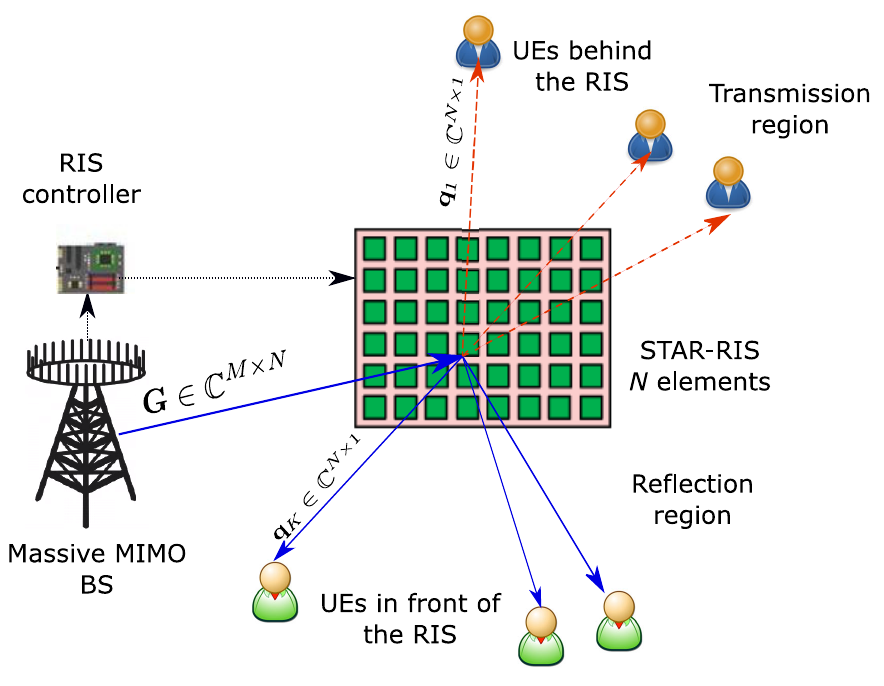}
			\caption{{ A mMIMO STAR-RIS assisted system with multiple UEs at transmission and reflection regions.  }}
			\label{Fig1}
		\end{center}
	\end{figure} 
	The STAR-RIS is able to configure  the transmitted ($ t $) and reflected ($ r $) signals by two independent coefficients. In particular, let $ t_{n} =( {\beta_{n}^{t}}e^{j \phi_{n}^{t}})s_{n}$ and $ r_{n}=( {\beta_{n}^{r}}e^{j \phi_{n}^{r}})s_{n} $ denote the transmitted 	and reflected signal by the $ n $th STAR-RIS element, respectively.\footnote{Note that here, we use $ \beta_{i}^{w_k} $, instead of $ \sqrt{\beta_{i}^{w_k}} $ as in \cite{Xu2021}, to denote the amplitude of the $ i $th RIS element in mode $ w_k $.  The reason for this change will become clear when  we present our proposed algorithm in Section \ref{PSConfig}.% 	Instead of using $  \sqrt{\beta} $ for the amplitudes of the RIS elements, we replace them with $ \beta $. Otherwise, their gradients will involve terms $ \frac{1}{\sqrt{\beta}} $, which may lead to numerical issues
	}  The amplitude and phase parameters  $ {\beta_{n}^{w_{k}}}\in [0,1] $ and $ \phi_{n}^{w_{k}} \in [0,2\pi)$, where the $  k $th UE can be in any of the two regions that  corresponds also to the  RIS mode, i.e.  transmission ($ t $) or reflection  ($ r $)
	\cite{Xu2021}, are independent.
	% \footnote{Recently, \cite{key} studied  the  correlation between the amplitude and the phase-shifts, which appears  in practical applications, and degrades the performance. The impact of this correlation left for future study.} 
	This model suggests that  $ \phi_{n}^{t} $ and $ \phi_{n}^{r} $ can be chosen independently, but  the choice of the amplitudes is based on the relationship expressed by the law of energy conservation as
	\begin{align}
		(\beta_{n}^{t})^{2}+(\beta_{n}^{r})^{2}=1,  \forall n \in \mathcal{N}.
	\end{align}
	
	Henceforth, for the sake of exposition, we denote $ \theta_{i}^{w_{k}}=e^{j\phi_{i}^{w_{k}}} $.

	\subsection{Operation Protocols}
	Our analysis is dedicated to the ES/MS protocols, which were presented in \cite{Mu2021}. Herein, we outline them by providing their main points.
	\subsubsection{ES protocol} All RIS elements serve simultaneously all UEs  in both   $ t $ and $ r $ regions. Especially, the  PBM for the $ k
	$th  UE is expressed as $ \bPhi_{w_{k}}^{\mathrm{ES}}=\diag( {\beta_{1}^{w_{k}}}\theta_{1}^{w_{k}}, \ldots,  {\beta_{N}^{w_{k}}}\theta_{N}^{w_{k}}) \in \mathbb{C}^{N\times N}$, where $ \beta_{n}^{w_{k}} \ge 0 $, $ 		(\beta_{n}^{t})^{2}+(\beta_{n}^{r})^{2}=1 $, and $ |\theta_{n}^{w_{k}}|=1, \forall n \in \mathcal{N} $.
	
	\subsubsection{MS protocol} The RIS elements are partitioned  into two groups of $ N_{t} $ and $ N_{r} $ elements that serve UEs in the  $  t$ and $ r $ regions, respectively. In other words,  $ N_{t}+N_{r}=N $. The PBM for $  k \in \mathcal{K}_{t} $ or $  k \in \mathcal{K}_{r} $
	is given by $ \bPhi_{w_{k}}^{\mathrm{MS}}=\diag( {\beta_{1}^{w_{k}}} \theta_{1}^{w_{k}}, \ldots,  {\beta_{N}^{w_k}}\theta_{N}^{w_{k}}) \in \mathbb{C}^{N\times N}$, where $ \beta_{n}^{w_{k}}\in \{0,1\}$, $ 	(\beta_{n}^{t})^{2}+(\beta_{n}^{r})^{2}=1 $, and $ |\theta_{i}^{w_{k}}|=1, \forall n \in \mathcal{N} $. As can be seen, this protocol is  a special case of the ES protocol, where the amplitude coefficients for transmission and reflection are restricted to binary values. As a result, the MS protocol is inferior of the ES counterpart since it cannot achieve the full-dimension transmission and reflection beamforming gain. Despite this performance degradation, it brings the advantage of lower computational complexity regarding the  PBM design.

	\subsection{Channel Model}\label{ChannelModel} 
	We assume  narrowband quasi-static block fading channels with each block having a duration of $\tau_{\mathrm{c}}$ channel uses. We adopt the standard time-division-duplex (TDD) protocol, which is preferable in mMIMO systems. Within TDD, we assume that each block  includes $\tau$ channel uses for the uplink training phase and $\tau_{\mathrm{c}}-\tau$ channel uses for the downlink data transmission phase. \textcolor{black}{Notably, contrary to other works, we aim to achieve  a unified analysis  regarding the channel estimation and data transmission phase that applies to a UE found in any of the $ t $ or $ r $ regions.}

	Let  $ \bG=[\bg_{1}\ldots,\bg_{N} ] \in \mathbb{C}^{M \times N}$ be the channel between the BS and the STAR-RIS with $ \bg_{i} \in \mathbb{C}^{M \times 1}$ for $ i\in \mathcal{N} $. Also,		$ \bq_{k} \in \mathbb{C}^{N \times 1}$
	%$ \bq_{t,k} \in \mathbb{C}^{N \times 1}$ and $ \bq_{r,k} \in \mathbb{C}^{N \times 1}$ 
	denotes the channel between the STAR-RIS and UE $ k $ that can be found on either side. The direct link between the BS  and UE  $ k $ is denoted as $ \bd_{k} $.  On this ground, we assume that all links are subject to correlated Rayleigh fading, which is normally the case in practice \cite{Bjoernson2020}.\footnote{\textcolor{black}{The consideration of correlated Rician fading, which includes an LoS component, is the topic of future work.}} In particular, we have
	\begin{align}
		\bG&=\sqrt{\tilde{ \beta}_{g}}\bR_{\mathrm{BS}}^{1/2}\bD\bR_{\mathrm{RIS}}^{1/2},\label{eq2}\\
		\bq_{k}&=\sqrt{\tilde{ \beta}_{k}}\bR_{\mathrm{RIS}}^{1/2}\bc_{k},\\	
		\bd_{k}&=\sqrt{\bar{ \beta}_{k}}\bR_{\mathrm{BS}}^{1/2}\bar{\bc}_{k},
	\end{align}where $ \bR_{\mathrm{BS}} \in \mathbb{C}^{M \times M} $ and $ \bR_{\mathrm{RIS}} \in \mathbb{C}^{N \times N} $, assumed to be known by the network, express the deterministic Hermitian-symmetric positive semi-definite correlation matrices at the BS and the RIS respectively.\footnote{Many previous works have assumed that the channel between the BS and the RIS is deterministic expressing a line-of-sight (LoS) component \cite{Nadeem2020,Papazafeiropoulos2021}, while the analysis here is more general since we assume that all links are correlated Rayleigh fading distributed. In particular, $ \bG $, as expressed in \eqref{eq2}  is based on the Kronecker channel model.}  Regarding $  \bR_{\mathrm{BS}} $, it can be modeled e.g., as in \cite{Hoydis2013}, and $ \bR_{\mathrm{RIS}} $ is modeled as in \cite{Bjoernson2020}. 	 Moreover, $\tilde{ \beta}_{g} $, $ \bar{ \beta}_{k} $, and $\tilde{ \beta}_{k} $  express the path-losses  of the BS-RIS, BS-UE $ k $, and RIS-UE $ k $ links in $ t $ or $r  $ region, respectively. Also, $ \mathrm{vec}(\bD)\sim \mathcal{CN}\left(\b0,\Id_{MN}\right) $, $ \bc_{k} \sim \mathcal{CN}\left(\b0,\Id_{N}\right) $, and $ \bar{\bc}_{k} \sim \mathcal{CN}\left(\b0,\Id_{N}\right) $ express the corresponding fast-fading components. 	
	
	\textcolor{black}{We note that the correlation matrices $ \bR_{\mathrm{RIS}} $ and $ \bR_{\mathrm{BS}} $ can be assumed to be known by the network since they can be obtained by existing estimation methods \cite{Neumann2018,Upadhya2018}. Alternatively, we can practically calculate the covariance matrices for both $ \bR_{\mathrm{RIS}} $ and $ \bR_{\mathrm{BS}}$, despite the fact that $ \bR_{\mathrm{RIS}} $ is passive. Especially, the expressions for these covariance matrices depend on the distances between the RIS elements and the BS antennas, respectively, as well as the angles between them. The distances are known from the construction of the RIS and the BS, and the angles can be calculated when the locations are given.  Hence, the covariance matrices can be considered to be known.}
	
	Given the PBM,  the aggregated channel vector for UE $ k $ \textcolor{black}{$ \bh_{k}=\bd_{k}+ \bG\bPhi_{w_{k}} \bq_{k} $} has a covariance matrix  $ \bR_{k}=\EE\{\bh_{k}\bh_{k}^{\H}\} $ given by
	\begin{align}
		\textcolor{black}{	\bR_{k}=\bar{ \beta}_{k}\bR_{\mathrm{BS}}+\hat{\beta}_{k}\tr(\bR_{\mathrm{RIS}} \bPhi_{w_{k}} \bR_{\mathrm{RIS}}  \bPhi_{w_{k}}^{\H})\bR_{\mathrm{BS}},\label{cov1}}
	\end{align}
	where we have used the independence between $ \bG $ and $ \bq_{k} $, $\hat{\beta}_{k}= \tilde{ \beta}_{g}\tilde{ \beta}_{k} $, $ \EE\{	\bq_{k}	\bq_{k}^{\H}\} =\tilde{ \beta}_{k} \bR_{\mathrm{RIS}}$, and  $ \EE\{\bV \bU\bV^{\H}\} =\tr (\bU) \Id_{M}$ with $\bU  $ being a deterministic square matrix, and $ \bV $ being any matrix with independent and identically distributed (i.i.d.) entries of zero mean and unit variance. Notably, when $\bR_{\mathrm{RIS}} =\Id_{N} $, $ \bR_{k} $ does not depend on the phase shifts but only on the amplitudes, as also observed in \cite{Papazafeiropoulos2022}.
	
	\textcolor{black}{\begin{remark}
			% In the case of independent Rayleigh fading, i.e., $ \bR_{\mathrm{RIS}}=\Id_{N} $ and $ \bR_{\mathrm{BS}}=\Id_{M} $, the variance of aggregated channel becomes $ \bR_{k}=\bar{ \beta}_{k}\Id_{M}+\hat{\beta}_{k}\sum_{i=1}^{N}(\beta_{i}^{w_{k}})^{2}\Id_{M} $, which is independent of the phase shifts. Below, it will be shown that if the analysis relies on statistical CSI, the achievable rate cannot be optimized with respect to the phase shifts under independent Rayleigh fading conditions. However, in practice, correlated fading is unavoidable, which allows the surface optimization in terms of the phase shifts.
			As shown in \eqref{cov1}, when independent Rayleigh fading is assumed, i.e., $\bR_{\mathrm{RIS}}=\Id_{N}$ and $\bR_{\mathrm{BS}}=\Id_{M}$, the covariance matrix of the aggregated channel becomes $\bR_{k}=(\bar{\beta}_{k}+\hat{\beta}_{k}\sum_{i=1}^{N}(\beta_{i}^{w_{k}})^{2})\Id_{M}$, which is independent of the phase shifts. 	 This reduces significantly the capability of the RIS in forming narrow beams and thus the performance is degraded accordingly.  Therefore, it is not possible to optimize the achievable rate with respect to the phase shifts under independent Rayleigh fading conditions.\footnote{\textcolor{black}{It is important to note that the recent works based on statistical CSI, such as \cite{VanChien2021,Papazafeiropoulos2021,Papazafeiropoulos2022,Zhao2020,Papazafeiropoulos2021b}, have shown a similar observation, i.e., in the case of no RIS correlation, the covariance matrix of the  aggregated channel $ \bR_{k} $  does not depend on the phase shifts.} } However, in practice, correlated fading is unavoidable, which enables the optimization of the surface in terms of the phase shifts.
	\end{remark}}

	\section{Channel Estimation}\label{ChannelEstimation}
	In practical systems, perfect CSI cannot be obtained. Especially, in mMIMO systems,  the TDD protocol is adopted and channels are estimated  by an uplink training phase with pilot symbols \cite{Bjoernson2017}. However, a RIS, being implemented by nearly passive elements without any RF chains, cannot process the estimated channels and obtain the received pilots by UEs. Also, it cannot transmit any pilot sequences to the BS for channel estimation.
	
	{\color{black}In general, there are two  approaches to channel estimation for RIS-aided communication systems, one focusing on the estimation of the individual channels such as \cite{Yang2020b,Wu2021,Zheng2022}, and the other obtaining the estimated aggregated channel such as  \cite{Nadeem2020, Papazafeiropoulos2021,Deshpande2022}. The first benefit of the latter approach is that its implementation  does not require any  extra hardware and power cost. Also, the estimated aggregated BS-RIS-user channel is sufficient for the transmission beamforming design 		for the RIS-related links. 	  It is easy to see that the  BS-RIS channel has a large dimension in the considered system, which results in a prohibitively high pilot overhead if individual channels need to be estimated.  This issue motivates us to employ the second approach in this paper, which has lower overhead and allows  estimated channels to be expressed in closed form. We will now provide the details of the adopted channel estimation method.}
	
	We assume that all UEs either in  $ t $ or $ r $ region send orthogonal pilot sequences. Specifically, we denote by $\bx_{k}=[x_{k,1}, \ldots, x_{k,\tau}]^{\H}\in \mathbb{C}^{\tau\times 1} $ the pilot sequence of UE $ k  $ that can be found in any of the two regions since the duration of the uplink training phase is $ \tau $ channel uses. Note that $ \bx_{k}^{\H}\bx_{l}=0~\forall k\ne l$ and $ \bx_{k}^{\H}\bx_{k}= \tau P$ joules with $ P =|x_{k,i}|^{2} ,~\forall k,i$, i.e., it is assumed that all UEs use the same average transmit power during the training phase.
	
	The received signal by the BS for  the whole uplink training period is written as
	\begin{align}
		\bY^{\tr}=\sum_{i=1}^{K}\bh_{i}\bx_{i}^{\H} +
		\bZ^{\tr},\label{train1}
	\end{align}
	where  $ \bZ^{\tr} \in \mathbb{C}^{M \times \tau} $ is the received AWGN matrix having independent columns with each one distributed as $ \mathcal{CN}\left(\b0,\sigma^2\Id_{M}\right)$. Obviously, in \eqref{train1}, there is a contribution from UEs of both regions.
	
	Multiplication of \eqref{train1} with the transmit training sequence from UE $ k $  removes the  interference  by other UEs that can be found in the same  or in the opposite region, and gives
	\begin{align}
		\br_{k}=\bh_{k}+\frac{\bz_{k}}{ \tau P},\label{train2}
	\end{align}
	where  $ \bz_{k}=\bZ^{\tr} \bx_{k}$.

	\begin{lemma}\label{PropositionDirectChannel}
		The LMMSE estimate of the aggregated channel $ \bh_{k} $ between the BS and  UE $ k $  is given by
		\begin{align}
			\hat{\bh}_{k}=\bR_{k}\bQ_{k} \br_{k},\label{estim1}
		\end{align}
		where $ \bQ_{k}\!=\! \left(\!\bR_{k}\!+\!\frac{\sigma^2}{ \tau P }\Id_{M}\!\right)^{\!-1}$, and $ \br_{k}$ is the noisy channel given by \eqref{train2}.
	\end{lemma}
	\begin{proof}
		Please see Appendix~\ref{lem1}.	
	\end{proof}
	
	The property of the orthogonality of LMMSE estimation gives the overall perfect channel in terms of the estimated channel $\hat{\bh}_{k}$ and estimation channel error vectors $\tilde{\bh}_{k}$ as
	\begin{align}
		\bh_{k}=\hat{\bh}_{k}+\tilde{\bh}_{k}\label{current}. \end{align}
	Both  $\hat{\bh}_{k}$ and $\tilde{\bh}_{k} $ have zero mean, and have variances  (cf. \eqref{var1}) 	\begin{align}
		\bPsi_{k}\!&=\!\bR_{k}\bQ_{k}\bR_{k},\label{Psiexpress}\\
		\tilde{\bPsi}_{k}&=\bR_{k}-\bPsi_{k},
	\end{align}respectively.  Given that $ \bh_{k} $ is not Gaussian, $\hat{\bh}_{k}$ and $\tilde{\bh}_{k}$ are not independent, but  they are uncorrelated and each of them has zero mean~\cite{Bjoernson2017}.
	
	{\color{black} It is clear from the above derivations that we  indeed follow a conventional channel estimation method from standard mMIMO systems, where only the aggregated instantaneous BS-UE channel $ \bh_{k} $ is estimated. In this way, the minimum pilot sequence length is $ K $, which is independent of the dimensions $ M $ and $ N $.  We note that if individual channels need to be estimated, the required complexity increases with $ M $ and $ N $, which are large in the considered system. Thus, the presented method of estimating aggregated channels offers significant overhead reduction. It is important to note that this channel estimation method is sufficient for the design of the precoder at the BS in statistical CSI-based approaches. 
		
		We also remark that the method of estimating aggregated channels presented above has not been previously applied to STAR-RIS-aided systems based on the two-timescale method, which adds to the novelty of our paper. Specifically, we have demonstrated that the same expression for  estimated channels can be used for users in either $t$ or $r$ regions, which has not been reported in previous papers  studying STAR-RIS. Note also that the proposed two-timescale transmission approach has a channel estimation phase that does not depend on $ N $, and thus,  it is applicable to both ES and MS protocols. However, the ES protocol has twice the number of optimizable variables and thus requires higher complexity compared to the MS protocol. An advantage of the two-timescale approach is that the surface needs to be  redesigned only when the statistical CSI changes. In contrast, instantaneous CSI-based schemes require beamforming calculations and information feedback in every channel coherence interval, leading to  high computational complexity, power consumption, and feedback overhead. For such schemes, the ES protocol is not practically appealing, and thus, our proposed two-timescale approach is certainly more viable.}

	\begin{remark}\label{rem3}
		Our analysis presented above relies on large-scale statistics for a given PBM, which is  obtained at every several coherence intervals. Thus, the optimization of the PBM that will be studied in the sequel is more practically appealing. Note that our method provides the estimated aggregated channel vector in closed-form. Other methods in the RIS literature such as \cite{He2019} do not result in analytical expressions, and  do not capture the correlation effect since they obtain the estimated channel per RIS element \cite{Nadeem2020}. Moreover, in the case of STAR-RIS, the only work on channel estimation is \cite{Wu2021} but it does not consider practical effects such as correlation and multiple antennas at the BS.
	\end{remark}
	
	\section{Downlink Data Transmission}\label{PerformanceAnalysis}
	The downlink data transmission from the BS to UE $ k $ in $ t  $ or $ r  $ region relies on TDD, which exploits channel reciprocity, i.e., the downlink channel equals the Hermitian transpose of the uplink channel. Hence, the received signal by UE $ k $ is expressed as
	\begin{align}
		r_{k}={\bh}^\H_{k}\bs+z_{k},\label{DLreceivedSignal}
	\end{align}
	where   $\bs=\sqrt{\lambda} \sum_{i=1}^{K}\sqrt{p_{i}}\bff_{i}l_{i}$ expresses the transmit signal vector  by the BS,  $ p_{i} $ is  the power allocated to UE $ i $,  and $ \lambda $ is a constant which is found such that $ \EE[\mathbf{s}^{\H}\mathbf{s}]=\rho $, where $ \rho $ is the total average power budget. Also, $z_{k} \sim \cC\cN(0,\sigma^{2})$ is the additive  white complex Gaussian noise at UE $k$. Moreover,   $\bff_{i} \in \bbC^{M \times 1}$ is the linear precoding vector and 
	% $x_{i,n}\in \sim \cC\cN(0,1)$ 
	$ l_{i} $ is     the corresponding data symbol with $ \EE\{|l_{i}|^{2}\}=1 $. In this paper we adopt equal power allocation among all UEs as usually happens in the mMIMO literature, i.e. $ p_i = \rho/K $ \cite{Hoydis2013}. Thus, $ \lambda $ is found to ensure $  \EE[\mathbf{s}^{\H}\mathbf{s}]=\rho$, which gives $ \lambda=\frac{K}{\EE\{\tr\bF\bF^{\H}\}} $,
	%	\begin{align}
		%		\EE\{\|\bs\|^{2}\}=\tr(\bP\bF^{\H}\bF)\le  P_{\mathrm{max}}\label{PowerConstraint},
		%	\end{align} 
	where $ \bF=[\bff_{1}, \ldots, \bff_{K}] \in\mathbb{C}^{M \times K}$.
	%	$ \bP=\diag(p_{1}, \ldots, p_{K})$
	
	According to the technique in~\cite{Medard2000} and by exploiting that UEs do not have instantaneous CSI but are aware of only statistical CSI, the received signal by UE $ k $  can be written as 
	\begin{align}
		r_{k}\!	&=\!\frac{\sqrt{\lambda\rho}}{K} (\EE\{\bh_{k}^{\H}\bff_{k}\}l_{k}	\!+{\bh}_{k}^{\H}\bff_{k}l_{k}-\EE\{\bh_{k}^{\H}\bff_{k}\}l_{k}\nn\\
		&+\!\sum_{i\ne k}^{K}\bh^\H_{k}\bff_{i}l_{i}Z\!+\!z_{k}\label{DLreceivedSignal1}.
	\end{align}
	
	Now, by using the use-and-then-forget bound \cite{Bjoernson2017}, which relies on the common assumption of the worst-case uncorrelated additive noise for the inter-user interference, we obtain a lower bound on downlink average SE in bps/Hz. We remark that this lower bound is tight for mMIMO as the number of antennas is very large.  Specifically, the achievable sum SE is given by
	\begin{align}
		\mathrm{SE}	=\frac{\tau_{\mathrm{c}}-\tau}{\tau_{\mathrm{c}}}\sum_{k=1}^{K}\log_{2}\left ( 1+\gamma_{k}\right)\!,\label{LowerBound}
	\end{align}
	where  $ \gamma_{k}$ is the downlink  signal-to-interference-plus-noise ratio (SINR), and  the pre-log fraction corresponds to  the percentage of samples per coherence block  for downlink data transmission. Note that  according to the use-and-forget bounding technique, the downlink SINR is given by
	\begin{align}
		\gamma_{k}=	\frac{	S_{k}}{{I}_{k}},\label{gamma1}
	\end{align}
	where 	
	\begin{align}
		{{S}}_{k}&=|\EE\{\bh_{k}^{\H}\bff_{k}\}|^{2} \label{Sig} \\
		{{I}}_{k}&=\EE\big\{ \big| {\bh}_{k}^{\H}\hat{\bh}_{k}-\EE\big\{
		{\bh}_{k}^{\H}\hat{\bh}_{k}\big\}\big|^{2}\big\}\!+\!\sum_{i\ne k}^{K}|\EE\{\bh^\H_{k}\bff_{i}\}|^{2}\!+\!\frac{K\sigma^{2}}{\rho \lambda}.\label{Int}\end{align}
	
	It is obvious that the final expressions for $ {{S}}_k $ and $ {{I}}_k $ depend on  the choice of the precoder and the derivation of the SINR. In this regard, we note that maximum ratio transmission (MRT) and  regularized zero-forcing (RZF) precoders are common options in the mMIMO literature. Herein, we select MRT for the sake of simplicity while RZF will be investigated in a future work. {\color{black} Since  aggregated channels between the BS and users involve the indirect link through the STAR-RIS, it is challenging to evaluate \eqref{Sig} and \eqref{Int} in closed-form.  In the following proposition we present tight approximations of \eqref{Sig} and \eqref{Int}, which are then used to optimize the phase shifts.} 
	\begin{proposition}\label{Proposition:DLSINR}
		Let  $ \mathbf{f}_k = \hat{\mathbf{h}}_k $, i.e., MRT precoding being used, then {\color{black} a tight approximation} of the downlink achievable SINR of UE $k$ for a given PBM $ \Phi_{w_k} $  in a STAR-RIS assisted mMIMO system, accounting for imperfect CSI, is given by
		%\footnote{\textcolor{black}{Proposition \ref{Proposition:DLSINR} provides a tight approximation as the  numerical results in Section \ref{Numerical} reveal. }}
		\begin{equation}
			\gamma_{k} \ {\color{black}\approx}\ 	\frac{S_{k}}{	\color{black}\tilde{I}_{k}},\label{gammaSINR}
		\end{equation}
		where
		\begin{align}
			{S}_{k}&=\tr^{2}\left(\bPsi_{k}\right)\!,\label{Num1}\\
			{\color{black}\tilde{I}_{k}}&=\sum_{i =1}^{K}\tr\!\left(\bR_{k}\bPsi_{i} \right)-\tr\left( \bPsi_{k}^{2}\right)+\frac{K\sigma^{2}}{ \rho}\sum_{i=1}^{K}\tr(\bPsi_{i}).\label{Den1}
		\end{align}
	\end{proposition} 
	\begin{proof}
		Please see Appendix~\ref{Proposition1}.	
	\end{proof}
	\section{Simultaneous Amplitudes and  Phase Shifts Configuration}\label{PSConfig}
	It is critical to find the PBM to optimize a performance measure of STAR-RIS assisted systems. In this paper, assuming infinite-resolution phase shifters, we formulate the optimization problem for maximizing the  sum SE with imperfect CSI and correlated fading. As mentioned in the preceding section, there are two operation protocols: ES protocol and MS protocol. In the following two subsections  we deal with these two protocols.
	\subsection{Optimization of {\color{black}Amplitudes} and Phase Shifts for ES protocol}
	For the ES protocol the formulated problem reads
	% \begin{align}\begin{split}
			% 		&\!\!\!\!\!(\mathcal{P}1)~\max_{\thetv,\betv} ~~	\mathrm{SE}\triangleq f(\thetv,\betv) \\
			% 		&~~~~~\mathrm{s.t}~~~~(\beta_{n}^{t})^{2}+(\beta_{n}^{r})^{2}=1,  \forall n \in \mathcal{N}\\
			% 		&~~~~~~~~~~~~\beta_{n}^{t}\ge 0, \beta_{n}^{r}\ge 0,|\theta_{n}^{t}|=|\theta_{n}^{r}|=1,~~  \forall n \in \mathcal{N}
			% 	\end{split}\label{Maximization} 
		% \end{align}
	\begin{equation}
		\begin{IEEEeqnarraybox}[][c]{rl}
			\max_{\thetv,\betv}&\quad {\color{black} f(\thetv,\betv) \triangleq \sum\nolimits_{k=1}^{K}\log_2(1+\frac{S_k}{\tilde{I}_k})}\\
			\mathrm{s.t}&\quad (\beta_{n}^{t})^{2}+(\beta_{n}^{r})^{2}=1,  \forall n \in \mathcal{N}\\
			&\quad\beta_{n}^{t}\ge 0, \beta_{n}^{r}\ge 0,~\forall n \in \mathcal{N}\\
			&\quad|\theta_{n}^{t}|=|\theta_{n}^{r}|=1, ~\forall n \in \mathcal{N}
		\end{IEEEeqnarraybox}\label{Maximization}\tag{$\mathcal{P}1$}
	\end{equation}
	where $\thetv=[(\thetv^{t})^{\T}, (\thetv^{r})^{\T}]^{\T} $ and $\betv=[(\betv^{t})^{\T}, (\betv^{r})^{\T}]^{\T} $. Note that to achieve a compact description we have vertically stacked $\thetv^{t}$ and $\thetv^{r}$ into  a single vector $\thetv$, and $\betv^{t}$ and $\betv^{r}$ into  a single vector $\betv$, respectively.  Also note that in the above problem formulation we have {\color{black}used the tight approximation of the SINR given in Proposition \ref{Proposition:DLSINR} to maximize the approximate sum SE, denoted by $f(\thetv,\betv)$}.
	For ease of exposition, we define two sets: $ \Theta=\{\thetv\ |\ |\theta_{i}^{t}|=|\theta_{i}^{r}|=1,i=1,2,\ldots N\} $, and $ \mathcal{B}=\{\betv\ |\ (\beta_{i}^{t})^{2}+(\beta_{i}^{r})^{2}=1,\beta_{i}^{t}\geq0,\beta_{i}^{r}\geq0,i=1,2,\ldots N\} $, which in fact together describe the feasible set of \eqref{Maximization}. Notably, the introduction of STAR-RIS imposes new challenges. In particular, the first constraint is not simple but includes the two types of passive beamforming, namely transmission and reflection beamforming, to be optimized, which are coupled with each other due to the energy conservation law.
	
	The  problem \eqref{Maximization} is non-convex and is coupled among the optimization variables, which are the amplitudes and the phase shifts for transmission and reflection. 	For the development of an efficient algorithm to solve  \eqref{Maximization} we remark that the sets $\Theta$ and $ \mathcal{B}$ are simple in the sense that their projection operators can be done in closed-form. This motivates us to apply the  projected  gradient ascent method (PGAM) \cite[Ch. 2]{Bertsekas1999} to optimize $\thetv$ and $\betv$, which is described next. \textcolor{black}{However, in the case of independent Rayleigh fading, $ \mathrm{SE} $ does not depend on $ \thetv $, which means that  optimization can take place only with respect to $ \betv $.}  
	% The optimization regarding the amplitudes and the phase shifts of the PBM can be performed simultaneously by applying . 
	
	% The convergence is guaranteed due to the transmit power constraint. Regarding the PGAM, we  denote $\mathrm{SE}=f(\thetv,\betv)$, . 
	
	The proposed PGAM consists of the following iterations
	\begin{subequations}\label{mainiteration}\begin{align}
			\thetv^{n+1}&=P_{\Theta}(\thetv^{n}+\mu_{n}\nabla_{\thetv}f(\thetv^{n},\betv^{n})),\label{step1} \\ \betv^{n+1}&=P_{\mathcal{B}}(\betv^{n}+{\mu}_{n}\nabla_{\betv}f(\thetv^{n},\betv^{n})).\label{step2} \end{align}
	\end{subequations}
	%  where $ \nabla_{\thetv}f(\thetv,\betv) $ and $ \nabla_{\betv}f(\thetv,\betv) $ correspond to the complex gradients of $ f(\thetv,\betv) $, which are defined as the partial derivatives of  $ f(\thetv,\betv) $ with respect to $ \thetv^{\ast} $ and $ \betv^{\ast} $. 
	In the above equations, the superscript denotes the iteration count. From the current iterate $(\thetv^{n},\betv^{n})$ we move along the gradient direction to increase the objective. In \eqref{mainiteration}, $\mu_n$ is the step size for both $\thetv$ and $\betv$. 
	% Compared to the traditional PGAM method where the step size for both $\thetv$ and $\betv$ is the same, i.e., $\alpha=1$, we introduce a constant $\alpha>0$ to adjust their step size. The reason for this is that the complex  gradients $ \nabla_{\thetv}f(\thetv,\betv) $ and $ \nabla_{\betv}f(\thetv,\betv) $ are of different scale, and thus using different step size may speed up the convergence.  
	%  and $ \mu_{n} $ and $ \bar{\mu}_{n} $ are the respective step sizes. 
	Also, in \eqref{mainiteration}, $P_{\Theta}(\cdot) $ and $ P_{\mathcal{B}}(\cdot) $ are the projections onto $ \Theta $ and $ \mathcal{B} $, respectively. 
	
	The choice of the step size in \eqref{step1} and \eqref{step2} is important to make the proposed PGAM converge. The ideal step size should be
	inversely proportional to the Lipschitz constant of the corresponding gradient but this is difficult to find for the considered problem. For this reason, we apply the Armijo-Goldstein backtracking line search to find the step size at each iteration. To this end, we define a quadratic approximation of $f(\thetv,\betv)$ as
	\begin{align}
		&	Q_{\mu}(\thetv, \betv;\bx,\by)=f(\thetv,\betv)+\langle	\nabla_{\thetv}f(\thetv,\betv),\bx-\thetv\rangle\nn\\
		&-\frac{1}{\mu}\|\bx-\thetv\|^{2}_{2}+\langle\nabla_{\betv}f(\thetv,\betv),\by-\betv\rangle-\frac{1}{\mu}\|\by-\betv\|^{2}_{2}.
	\end{align}    
	Note that in this paper we define $\langle\mathbf{x},\mathbf{y}\rangle=2\Re{\mathbf{x}^{\H}\mathbf{y}}$ for complex-valued $\mathbf{x}$ and ${\mathbf{y}}$ and $\langle\mathbf{x},\mathbf{y}\rangle=\mathbf{x}^{\T}\mathbf{y}$ for non complex-valued $\mathbf{x}$ and ${\mathbf{y}}$. 
	Now, let $ L_n>0 $, and $ \kappa \in (0,1) $. Then the step size $  \mu_{n} $ in \eqref{mainiteration} can be found as $ \mu_{n} = L_{n}\kappa^{m_{n}} $, where $ m_{n} $ is the
	smallest nonnegative integer satisfying
	\begin{align}
		f(\thetv^{n+1},\betv^{n+1})\geq	Q_{L_{n}\kappa^{m_{n}}}(\thetv^{n}, \betv^{n};\thetv^{n+1},{\color{black}\betv^{n+1}}),
	\end{align}
	which can be done by an iterative procedure. In the proposed PGAM, we use the step size at iteration $n$ as the initial step size at iteration $n+1$. The proposed PGAM is summarized in Algorithm \ref{Algoa1}. 
	\begin{algorithm}[th]
		\caption{Projected Gradient Ascent Method for the RIS Design\label{Algoa1}}
		\begin{algorithmic}[1]
			\STATE Input: $\thetv^{0},\betv^{0},\mu_{1}>0$, $\kappa\in(0,1)$
			\STATE $n\gets1$
			\REPEAT
			\REPEAT \label{ls:start}
			\STATE $\thetv^{n+1}=P_{\Theta}(\thetv^{n}+\mu_{n}\nabla_{\thetv}f(\thetv^{n},\betv^{n}))$
			\STATE $\betv^{n+1}=P_{B}(\betv^{n}+\mu_{n}\nabla_{\betv}f(\thetv^{n},\betv^{n}))$
			\IF{ $f(\thetv^{n+1},\betv^{n+1})\leq Q_{\mu_{n}}(\thetv^{n},\betv^{n};\thetv^{n+1},{\color{black}\betv^{n+1}})$}
			\STATE $\mu_{n}=\mu_{n}\kappa$
			\ENDIF
			\UNTIL{ $f(\thetv^{n+1},\betv^{n+1})>Q_{\mu_{n}}(\thetv^{n},\betv^{n};\thetv^{n+1},{\color{black}\betv^{n+1}})$}\label{ls:end}
			\STATE $\mu_{n+1}\leftarrow\mu_{n}$
			\STATE $n\leftarrow n+1$
			\UNTIL{ convergence}
			\STATE Output: $\thetv^{n+1},\betv^{n+1}$
		\end{algorithmic}
	\end{algorithm} 
	
	We present the complex-valued gradients in the following lemma.
	\begin{lemma}\label{LemmaGradients}
		The complex gradients $ \nabla_{\thetv}f(\thetv,\betv) $ and  $\nabla_{\betv}f(\thetv,\betv) $ are given in closed-forms by
		% 	 and $ \betv^{*} $ are given by
		\begin{subequations}
			\begin{align}
				\nabla_{\thetv}f(\thetv,\betv) &=[\nabla_{\thetv^{t}}f(\thetv,\betv)^{\T}, \nabla_{\thetv^{r}}f(\thetv,\betv)^{\T}]^{\T},\\
				\nabla_{\thetv^{t}}f(\thetv,\betv)&=\frac{\tau_{\mathrm{c}}-\tau}{\tau_{\mathrm{c}}\log2}\sum_{k=1}^{K}\frac{	{\color{black}\tilde{I}_k}\nabla_{\thetv^{t}}{S_{k}}-S_{k}	\nabla_{\thetv^{t}}{{\color{black}\tilde{I}_k}}}{(1+\gamma_{k}){\color{black}\tilde{I}_k}^{2}} ,\\
				\nabla_{\thetv^{r}}f(\thetv,\betv)&=\frac{\tau_{\mathrm{c}}-\tau}{\tau_{\mathrm{c}}\log2}\sum_{k=1}^{K}\frac{	{\color{black}\tilde{I}_k}\nabla_{\thetv^{r}}{S_{k}}-S_{k}	\nabla_{\thetv^{r}}{{\color{black}\tilde{I}_k}}}{(1+\gamma_{k}){\color{black}\tilde{I}_k}^{2}}, 
			\end{align}
		\end{subequations}
		
		% 		\nabla_{\betv}f(\thetv^{(n)},\betv^{(n)})&=
		
		where
		%  $ 	\nabla_{\bu}{S_{k}} $ and $ 	\nabla_{\bu}{I_{k}} $ are given by \eqref{sk} and \eqref{ik}, respectively.
		% \begin{align}
			% \nabla_{\bu}{S_{k}}&=v_{k}\diag(\bA_{k}\diag{\bar{\bu}})\nn\\
			% &\times 	\tr\left((\bQ_{k}\bR_{k}+\bR_{k}\bQ_{k}-\bQ_{k}\bR_{k}^{2}\bQ_{k})\bR_{\mathrm{BS}}\right),\label{derivS}\\
			% 	&\!\!\!	\nabla_{\bu}{I_{k}}=\diag(\tilde{\bA}_{k}\diag{\bar{\bu}}),\label{derivI}
			% \end{align}
		\begin{subequations}
			\begin{align}
				\nabla_{\thetv^{t}}S_{k}&=\begin{cases}
					\nu_{k}\diag\bigl(\mathbf{A}_{t}\diag(\boldsymbol{{\beta}}^{t})\bigr) & w_{k}=t\\
					0 & w_{k}=r
				\end{cases}\label{derivtheta_t}\\
				\nabla_{\thetv^{r}}S_{k}&=\begin{cases}
					\nu_{k}\diag\bigl(\mathbf{A}_{r}\diag(\boldsymbol{{\beta}}^{r})\bigr) & w_{k}=r\\
					0 & w_{k}=t
				\end{cases}\label{derivtheta_r}\\
				\nabla_{\thetv^{t}}{\color{black}\tilde{I}_k} &=\diag\bigl(\tilde{\mathbf{A}}_{kt}\diag(\boldsymbol{{\beta}}^{t})\bigr)\label{derivtheta_t_Ik}\\
				\nabla_{\thetv^{r}}{\color{black}\tilde{I}_k} &=\diag\bigl(\tilde{\mathbf{A}}_{kr}\diag(\boldsymbol{\beta}^{r})\bigr)\label{derivtheta_r_Ik}
			\end{align}
		\end{subequations}
		with $\bA_{w_k}= \bR_{\mathrm{RIS}} \bPhi_{w_k} \bR_{\mathrm{RIS}} $ for $w_k\in\{t,r\}$, $ \textcolor{black}{\nu_{k}}=2\hat{\beta}_{k}\tr\left(\bPsi_{k}\right)\tr((\bQ_{k}\bR_{k}+\bR_{k}\bQ_{k}-\bQ_{k}\bR_{k}\bQ_{k})\bR_{\mathrm{BS}}) $, 
		% $ \tilde{\mathbf{A}}_{ku}=\bar{\nu}_{k}\mathbf{A}_{u}+\sum\nolimits _{i\in \mathcal{K}_u}^{K}\tilde{\nu}_{ki}\mathbf{A}_{i}$,
		\begin{equation}
			\tilde{\mathbf{A}}_{ku}=\begin{cases}
				\bigl(\bar{\nu}_{k}+\sum\nolimits _{i\in\mathcal{K}_{u}}^{K}\tilde{\nu}_{ki}\bigr)\mathbf{A}_{u} & w_{k}=u\\
				\sum\nolimits _{i\in\mathcal{K}_{u}}^{K}\tilde{\nu}_{ki}\mathbf{A}_{u} & w_{k}\neq u,
			\end{cases}\label{A_tilde_general}
		\end{equation}
		$u\in\{t,r\}$, $\bar{\nu}_{k}=\hat{\beta}_{k}\tr\bigl(\check{\boldsymbol{\Psi}}_{k}\mathbf{R}_{\mathrm{BS}}\bigr)$,
		$\tilde{\nu}_{ki}=\hat{\beta}_{k}\tr\bigl(\tilde{\mathbf{R}}_{ki}\mathbf{R}_{\mathrm{BS}}\bigr)$,
		% 	and $
		% $ \tilde{\bA}_{k} =\bar{v}_{k}\bA_{k}+\sum_{i=1}^{K}\tilde{v}_{ki}\bA_{i}$, $ \bar{v}_{k}=\hat{\beta}_{k}\tr(\tilde{\bPsi}_{k}\bR_{\mathrm{BS}}) $,  $ \tilde{v}_{ki}=\hat{\beta}_{k}\tr(\tilde{\bR}_{ki}\bR_{\mathrm{BS}}) $, 
		$ \check{\bPsi}_{k}={\bPsi}-2(\bQ_{k}\bR_{k}\bPsi_{k}+\bPsi_{k}\bR_{k}\bQ_{k}-\bQ_{k}\bR_{k}\bPsi_{k}\bR_{k}\bQ_{k}) $, $ \bPsi=\sum_{i=1}^{K}\bPsi_{i}$, $\tilde{\mathbf{R}}_{ki}=\mathbf{Q}_{i}\mathbf{R}_{i}\bar{\mathbf{R}}_{k}-\mathbf{Q}_{i}\mathbf{R}_{i}\bar{\mathbf{R}}_{k}\mathbf{R}_{i}\mathbf{Q}_{i}+\bar{\mathbf{R}}_{k}\mathbf{R}_{i}\mathbf{Q}_{i}$, and $\bar{\mathbf{R}}_{k}=\mathbf{R}_{k}+\frac{K\sigma^{2}}{\rho}\mathbf{I}_{M}$.
		Similarly, the  gradient $\nabla_{\betv}f(\thetv,\betv) $ is given by
		\begin{subequations}\label{eq:deriv:wholebeta}
			\begin{align}
				\nabla_{\betv}f(\thetv,\betv) &=[\nabla_{\betv^{t}}f(\thetv,\betv)^{\T}, \nabla_{\betv^{r}}f(\thetv,\betv)^{\T}]^{\T},\\
				\nabla_{\betv^{t}}f(\thetv,\betv)&=\frac{\tau_{\mathrm{c}}-\tau}{\tau_{\mathrm{c}}\log2}\sum_{k=1}^{K}\frac{	{\color{black}\tilde{I}_{k}}\nabla_{\betv^{t}}{S_{k}}-S_{k}	\nabla_{\betv^{t}}{{\color{black}\tilde{I}_{k}}}}{(1+\gamma_{k}){\color{black}\tilde{I}_{k}}^{2}}, \label{gradbetat:final}\\
				\nabla_{\betv^{r}}f(\thetv,\betv)&=\frac{\tau_{\mathrm{c}}-\tau}{\tau_{\mathrm{c}}\log2}\sum_{k=1}^{K}\frac{	{\color{black}\tilde{I}_{k}}\nabla_{\betv^{r}}{S_{k}}-S_{k}	\nabla_{\betv^{r}}{{\color{black}\tilde{I}_{k}}}}{(1+\gamma_{k}){\color{black}\tilde{I}_{k}}^{2}} ,
			\end{align}
		\end{subequations} 
		where
		\begin{subequations}
			\begin{align}
				\nabla_{\betv^{t}}S_{k}&=\begin{cases}
					2\nu_k\Re\bigl\{\diag\bigl(\mathbf{A}_{k}\herm\diag(\btheta^{t})\bigr)\bigr\} & w_{k}=t\\
					0 & w_{k}=r
				\end{cases}\label{derivbeta_t}\\
				\nabla_{\betv^{r}}S_{k}&=\begin{cases}
					2\nu_k\Re\bigl\{\diag\bigl(\mathbf{A}_{k}\herm\diag(\btheta^{r})\bigr)\bigr\} & w_{k}=r\\
					0 & w_{k}=t
				\end{cases}\label{derivbeta_r}\\
				\nabla_{\betv^{t}}{\color{black}\tilde{I}_k} &=2\Re\bigl\{\diag\bigl(\tilde{\mathbf{A}}_{kt}\herm\diag(\boldsymbol{\btheta}^{t})\bigr)\bigr\}\\
				\nabla_{\betv^{r}}{\color{black}\tilde{I}_k} &=2\Re\bigl\{\diag\bigl(\tilde{\mathbf{A}}_{kr}\herm\diag(\boldsymbol{\btheta}^{r})\bigr)\bigr\}.
			\end{align}
		\end{subequations}
		Note that $\nabla_{\betv}f(\thetv,\betv)$ is real-valued.\end{lemma}
	\begin{proof}
		Please see Appendix~\ref{lem2}.	
	\end{proof}
	\begin{remark}
		As mentioned earlier we use $ \beta_{i}^{w_k} $, instead of $ \sqrt{\beta_{i}^{w_k}} $ as in \cite{Xu2021}, to denote the amplitude of the $ i $th RIS element in mode $ w_k$. The purpose of that maneuver  is now clear. In fact, if $ \sqrt{\beta_{i}^{w_k}} $ were used to represent the amplitude, then the gradient $\nabla_{\betv}f(\thetv,\betv)$  would be similar to \eqref{eq:deriv:wholebeta} but contain the term  $ \sqrt{\beta_{i}^{w_k}} $ in the denominator. This will make $\nabla_{\betv}f(\thetv,\betv)$  ill-conditioned (i.e., extremely large), which in turn can cause numerical issues in the execution of Algorithm \ref{Algoa1} in practice.
	\end{remark}
	
	To conclude the description of Algorithm \ref{Algoa1}, we now provide the projection onto the sets  $ \Theta $ and $ \mathcal{B} $. First, it is straightforward to check that, for a given $\thetv\in \mathbb{C}^{2N\times 1}$  $P_{\Theta}(\thetv)$ is given by 
	\begin{equation}
		P_{\Theta}(\thetv)=\thetv/|\thetv|=e^{j\angle\thetv},
	\end{equation}
	where the operations in the right-hand side of the above equation are performed entrywise. 
	
	The projection  $P_{ \mathcal{B} }(\betv)$ deserves special attention. Note that the constraint $(\beta_{i}^{t})^{2}+(\beta_{i}^{r})^{2}=1,\beta_{i}^{t}\geq0,\beta_{i}^{r}\geq0$ indeed defines the first quadrant of the  unit circle. Thus, the expression of the projection onto $\mathcal{B}$ is rather complicated. To make $P_{\mathcal{B} }(\betv)$ more efficient we allow $\beta_{i}^{t}$ and $\beta_{i}^{r}$ to take negative value during the iterative process. However, we remark that this step \emph{does not} affect the optimality of the proposed solution since  we can change the sign of both $\beta_i^{u}$ and $\theta_{i}^{u}$, $u\in{t,r}$ and still achieve the same objective. As a result, we can project $\beta_{i}^{t}$ and $\beta_{i}^{r}$ onto the entire unit circle, and thus  we can write $P_{ \mathcal{B} }(\betv)$  as
	\begin{subequations}
		\begin{align}
			\left[\ensuremath{P_{\mathcal{B}}(}\boldsymbol{\beta})\right]_{i} & =\frac{{\beta}_{i}}{\sqrt{{\beta}_{i}^{2}+{\beta}_{i+N}^{2}}},i=1,2,\ldots,N\\
			\left[\ensuremath{P_{\mathcal{B}}(}\boldsymbol{\beta})\right]_{i+N} & =\frac{{\beta}_{i+N}}{\sqrt{{\beta}_{i}^{2}+{\beta}_{i+N}^{2}}}, i=1,2,\ldots,N.
		\end{align}
	\end{subequations}
	\subsubsection*{Complexity Analysis of Algorithm \ref{Algoa1}}
	We now provide the complexity analysis for each iteration of Algorithm \ref{Algoa1}  in terms of the required number of complex multiplications using big-O notation, which is in particular relevant for large $M$ and $N$ as considered in this paper. We note Algorithm \ref{Algoa1} only requires the first-order information, i.e., the objective and its gradient value. Let us analyze the complexity of computing the objective value. First, we need to compute $\bR_{k}$ which can be written as  $\bR_{k}=\bar{ \beta}_{k}\bR_{\mathrm{BS}}+ \hat{\beta}_{k}\tr(\bR_{\mathrm{RIS}} \bPhi_{w_{k}} \bR_{\mathrm{RIS}}  \bPhi_{w_{k}}^{\H})\bR_{\mathrm{BS}}=\bar{ \beta}_{k}\bR_{\mathrm{BS}}+ \hat{\beta}_{k}\tr(\bA_{w_k} \bPhi_{w_{k}}^{\H})\bR_{\mathrm{BS}}$. Now it is obvious that we need to obtain the term $\tr(\bA_{w_k} \bPhi_{w_{k}}^{\H})$. To this end, we note that since $\bPhi_{w_{k}}$ is diagonal, \emph{only diagonal elements of $\bA_{w_k} $, i.e., $\diag(\bA_{w_k})$, are required}. Next, computing $\bR_{\mathrm{RIS}} \bPhi_{w_{k}}$ requires $N^2$ complex multiplications since $\bPhi_{w_{k}}$ is diagonal, and thus, $\diag(\bA_{w_k})=\diag(\bR_{\mathrm{RIS}} \bPhi_{w_{k}} \bR_{\mathrm{RIS}})$ requires $O(N^2)$ complex multiplications. As a result, the complexity to compute  $\tr(\bA_{w_k} \bPhi_{w_{k}}^{\H})$ is $O(N^2+N)$, and thus, the complexity to obtain $\bR_{k}$ is $O(N^2+N+M^2)$ since $O(M^2)$ additional complexity multiplications are required to obtain $\tr(\bA_{w_k} \bPhi_{w_{k}}^{\H})\bR_{\mathrm{BS}}$. Recall that $\bPsi_{k}=\bR_{k}\bQ_{k}\bR_{k}=\bR_{k}\bigl(\bR_{k}+\frac{\sigma^2}{ \tau P }\Id_{M}\bigr)^{-1}\bR_{k}$ and thus it would take $O(M^3)$ to compute it, which is due to the calculation of the involving matrix inversion. Here we present a more efficient way to compute $\bPsi_{k}$. Let $\bR_{\mathrm{BS}}=\bU\bSigma\bU^{\H}$, where $\bSigma$ is diagonal and $\bU$ is unitary, be the eigenvalue decomposition (EVD) of $\bR_{\mathrm{BS}}$ and $\alpha_{k}=\hat{\beta}_{k}\tr(\bA_{w_k} \bPhi_{w_{k}}^{\H})$. We remark that \emph{the EVD of $\bSigma$ is only performed once} before Algorithm \ref{Algoa1} is executed. Then, we can write
	\begin{align}\label{eq:Qk}
		\bQ_{k}&=\bigl(\bR_{k}+\frac{\sigma^2}{ \tau P }\Id_{M}\bigr)^{-1}
		\nn\\
		&=\bigl(\alpha_{k}\bU\bSigma\bU^{\H}+\frac{\sigma^2}{ \tau P }\Id_{M}\bigr)^{-1}
		\nn\\
	&	=\bU\bigl(\alpha_{k}\bSigma+\frac{\sigma^2}{ \tau P }\Id_{M}\bigr)^{-1}\bU^{\H},
	\end{align} 
	where we have used the fact that $\bR_{k}=\alpha_{k}\bR_{\mathrm{BS}}=\alpha_{k}\bU\bSigma\bU^{\H}$.
	% 	In this way the complexity of computing $\bQ_k$ can be reduced to $O(M^2)$ since $\alpha_{k}\bSigma+\frac{\sigma^2}{ \tau P }\Id_{M}$ is diagonal. 
	Substituting \eqref{eq:Qk} into \eqref{Psiexpress}, we immediately have
	\begin{align}
		\bPsi_{k}&=\alpha_{k}^2\bU\bSigma\bigl(\alpha_{k}\bSigma+\frac{\sigma^2}{ \tau P }\Id_{M}\bigr)^{-1}\bSigma\bU^{\H} 	\nn\\
		&=\bU\bSigma\bigl(\alpha_{k}^{-1}\bSigma+\frac{\sigma^2}{ \tau P \alpha_{k}^2}\Id_{M}\bigr)^{-1}\bSigma\bU^{\H}
			\nn\\
		& =\bU\bSigma\bar{\bSigma}_k\bSigma\bU^{\H},\end{align} 
	where $\bar{\bSigma}_k=\bigl(\alpha_{k}^{-1}\bSigma+\frac{\sigma^2}{ \tau P \alpha_{k}^2}\Id_{M}\bigr)^{-1}$. Note that $\bar{\bSigma}_k$ is diagonal and takes $O(M)$ complex multiplications to compute. Now it is clear that $\tr(\bPsi_{k})=\tr(\bSigma\bar{\bSigma}_k\bSigma)$ requires $O(M)$ complex multiplications  to compute, which is indeed the complexity to compute $S_k$ in \eqref{Num1}. To compute $I_k$ we have
	\begin{align}
\!\!\!\!\!		\bPsi&=\sum\nolimits_{i=1}^{K}\!\!\bPsi_{i}
	\nn\\
&=\bU\bSigma\Bigl(\sum\nolimits_{i=1}^{K}\!\!\bar{\bSigma_i}\Bigr)\bSigma\bU^{\H}	\nn\\
&=\bU\bSigma\bar{\bSigma}\bSigma\bU^{\H},
	\end{align}where $\bar{\bSigma}=\sum\nolimits_{i=1}^{K}\bar{\bSigma_i}$. We note that obtaining $\bar{\bSigma}$ once all $\bar{\bSigma_i}$'s are known requires only $KM$ complex additions, which is negligible.
	Thus, it follows that 
	\begin{align}
		\sum\nolimits_{i=1}^{K}\tr(\bR_{k}\bPsi_{i}) 	&= \tr(\bR_{k}\bPsi)	\nn\\
		&=\alpha_{k}\tr(\bU\bSigma^2\bar{\bSigma}\bSigma\bU^{\H})	\nn\\
		&=\alpha_{k}\tr(\bSigma^2\bar{\bSigma}\bSigma),
	\end{align}
	and that
	\begin{equation}
		\sum\nolimits_{i=1}^{K}\tr(\bPsi_{i}) = \tr(\bPsi) = \tr(\bSigma\bar{\bSigma}\bSigma).
	\end{equation}
	Summarizing the above results, we can conclude that the complexity to compute $f(\thetv,\betv)$ is $O(K(N^2+M^2))$.
	
	Next we present the complexity to compute $\nabla_{\thetv}f(\thetv,\betv)$ and $\nabla_{\betv}f(\thetv,\betv)$. Recall that  $v_{k}$ in \eqref{derivtheta_t} and \eqref{derivtheta_r} is given by $v_{k}=2\hat{\beta}_{k}\tr\left(\bPsi_{k}\right)\tr((\bQ_{k}\bR_{k}+\bR_{k}\bQ_{k}-\bQ_{k}\bR_{k}\bQ_{k})\bR_{\mathrm{BS}})$. Following the above analysis, we can write $\bQ_{k}\bR_{k}\bR_{\mathrm{BS}}$ as
	$\bQ_{k}\bR_{k}\bR_{\mathrm{BS}}=\bR_{k}\bQ_{k}=\alpha_{k}\bU(\alpha_{k}\bSigma+ \frac{\sigma^2}{ \tau P \alpha_{k}^2}\Id_{M})^{-1}\bSigma^{2}\bU^{\H}$ and $\bQ_{k}\bR_{k}\bQ_{k}\bR_{\mathrm{BS}}$ as $\bQ_{k}\bR_{k}\bQ_{k})\bR_{\mathrm{BS}}=\alpha_{k}\bU(\alpha_{k}\bSigma+ \frac{\sigma^2}{ \tau P \alpha_{k}^2}\Id_{M})^{-2}\bSigma^{2}\bU^{\H}$. Thus, we can rewrite $v_{k}$ equivalently as
	\begin{align}
	&	v_{k} = 2\hat{\beta}_{k}\alpha_{k}\tr\left(\bPsi_{k}\right)
			\nn\\
		&\times\Bigl(2\tr\bigl(\bigl(\alpha_{k}\bSigma+ \frac{\sigma^2}{ \tau P \alpha_{k}^2}\Id_{M}\bigr)^{-1}\bSigma^{2}\bigr)\nn\\
		&-\tr\bigl(\bigl(\alpha_{k}\bSigma+ \frac{\sigma^2}{ \tau P \alpha_{k}^2}\Id_{M}\bigr)^{-2}\bSigma^{2}\bigr)\Bigr),
	\end{align} which requires  $O(M)$ complex multiplications to obtain since $\tr\left(\bPsi_{k}\right)$ is already computed and the involving matrices in the second term of the above equation are diagonal.
	
	Next, to obtain $\nabla_{\thetv^{t}}S_{k}$ in \eqref{derivtheta_t}, we further need to calculate the diagonal elements of $\mathbf{A}_{t}\diag(\boldsymbol{{\beta}}^{t})$, which can be obtained by multiplying each diagonal element of $\mathbf{A}_{t}$ with the corresponding entry of $\boldsymbol{{\beta}}^{t}$, i.e., $\diag\bigl(\mathbf{A}_{t}\diag(\boldsymbol{{\beta}}^{t})\bigr)=\diag\bigl(\mathbf{A}_{t}\bigr)\odot \betv^{t}$,  where $\odot$ represents the entry-wise multiplication. We remark that the term $\diag\bigl(\mathbf{A}_{t}\bigr)$ is already computed when calculating $\bR_{k}$, and  thus, the complexity to compute $\nabla_{\thetv^{t}}S_{k}$ is  $O(N)$. {\color{black}Apparently, the same complexity is  required to obtain $\nabla_{\thetv^{r}}S_{k}$ in \eqref{derivtheta_r}.} The complexity of calculating $\nabla_{\thetv^{u}}I_{k}$, $u\in\{t,r\}$ follows  similar lines. More specifically, both $\bar{\nu}_{k}$ and $\sum\nolimits _{i\in\mathcal{K}_{u}}^{K}\tilde{\nu}_{ki}$ require $O(M^2)$ to compute. In summary, the complexity of computing  the gradients for each iteration is $O(K(N^2+M^2))$. 
	{\color{black}
		\subsubsection*{Convergence Analysis of Algorithm \ref{Algoa1}}
		The convergence of Algorithm \ref{Algoa1} is guaranteed by following standard arguments for projected gradient methods. First, the gradients $\nabla_{\thetv}f(\thetv,\betv)$ and $\nabla_{\betv}f(\thetv,\betv)$ are Lipschitz continuous\footnote{\color{black}A function $\bh(\bx)   $ is said to be Lipschitz continuous over the set $D$ if there exists $L>0$ such that $||\bh(\bx)-\bh(\by)  ||\leq L||\bx-\by||_2$} over the feasible set as they comprise basic functions as given above. Let $L_{\thetv }$ and $L_{\betv}$ be the Lipschitz constant of $\nabla_{\thetv}f(\thetv,\betv)$ and $\nabla_{\betv}f(\thetv,\betv)$, respectively. Then it holds that \cite[Chapter 2]{Bertsekas1999}
		\begin{align}
			f(\bx,\by) &\geq f(\thetv,\betv)
			+\langle	\nabla_{\thetv}f(\thetv,\betv),\bx-\thetv\rangle-\frac{1}{L_{\thetv }}\|\bx-\thetv\|^{2}_{2}
			\nn\\
			&\quad\quad+\langle\nabla_{\betv}f(\thetv,\betv),\by-\betv\rangle-\frac{1}{L_{\betv}}\|\by-\betv\|^{2}_{2}\nn\\
			&\geq f(\thetv,\betv)
			+\langle	\nabla_{\thetv}f(\thetv,\betv),\bx-\thetv\rangle-\frac{1}{L_{\max }}\|\bx-\thetv\|^{2}_{2}\nn
			\\ \nn&\quad\quad+\langle\nabla_{\betv}f(\thetv,\betv),\by-\betv\rangle-\frac{1}{L_{\max}}\|\by-\betv\|^{2}_{2}		
		\end{align}
		where $L_{\max}=\max(L_{\thetv },L_{\betv})$. Thus, the line search procedure of Algorithm \ref{Algoa1} (i.e. the loop between Steps \ref{ls:start} -- \ref{ls:end}) terminates in finite iterations since the condition in Step \ref{ls:end} must be satisfied when  $\mu_n <L_{\max}$. More specifically, given $\mu_{n-1}$, the maximum number of steps in the line search procedure is $\left\lceil \frac{\log(L_{\max}\mu_{n-1})}{\log\kappa}\right\rceil $, where $\log()$ denotes the natural logarithm and 
		$\left\lceil \cdot \right\rceil$ 
		denotes the smallest integer that is larger than or equal to the argument. Also, due to the line search we automatically have  an increasing sequence of objectives, i.e., $f(\thetv^{n+1},\betv^{n+1})\geq f(\thetv^{n},\betv^{n})$. Since the feasible sets $\Theta$  and $\mathcal{B}$ are compact, $f(\thetv^{n},\betv^{n})$ must converge. However, we remark that Algorithm \ref{Algoa1} is only guaranteed to converge to a stationary point of \eqref{Maximization}, which is not necessarily an optimal solution due to the nonconvexity of \eqref{Maximization}. We also note that $L_{\thetv }$ and $L_{\betv}$ are not required to run Algorithm \ref{Algoa1}.
		\begin{remark}
			A  question naturally arising is why we have optimized the amplitude, $[\betv_{u}]_i$, and phase shift, $[\thetv_{u}]_i$, $u\in \{t,r\}$ separately, rather than optimizing them as a single complex, e.g, $[\varv_{u}]_i=[\betv_{u}]_ie^{j[\thetv_{u}]_i}$. \textcolor{black}{The latter would certainly make the presentation of the proposed method more elegant}. However, interestingly enough, we find by extensive numerical experiments that both ways give the same performance in many cases. {\color{black}We note that this does not mean the amplitudes corresponding to the reflection or transmission mode of  most of the STAR-IRS elements are  close to 1.  There is indeed  a significant gap between the ES and MS mode as shown in the next section.} However, in some cases, using two separate variables yields a better performance. This numerical observation has led to the current presentation of the proposed method where  amplitudes and phase shifts are optimized separately.
		\end{remark}
	}
\subsection{Optimization of \textcolor{black}{Amplitudes} and Phase Shifts for MS protocol}
In the case of the MS scheme, the values of amplitude are forced to be binary, i.e.,  $\betv^{t}_n\in \{ 0, 1\} $ and $\betv^{r}_n\in \{ 0, 1\} $. Thus, the optimization problem for the MS protocol is stated as
\begin{equation}
	\begin{IEEEeqnarraybox}[][c]{rl}
		\max_{\thetv,\betv}&\quad f(\thetv,\betv)\\
		\mathrm{s.t}&\quad \beta_{n}^{t}+\beta_{n}^{r}=1,  \forall n \in \mathcal{N}\\
		&\quad\beta_{n}^{t}\in \{ 0, 1\}, \beta_{n}^{r}\in \{ 0, 1\},~ \forall n \in \mathcal{N}\\
		&\quad |\theta_{n}^{t}|=|\theta_{n}^{r}|=1, ~ \forall n \in \mathcal{N}.
	\end{IEEEeqnarraybox}\label{MSMaximization}\tag{$\mathcal{P}2$}
\end{equation}

The binary constraints on $\betv^{t}_n$ and $\betv^{r}_n$ in \eqref{MSMaximization} make it far more difficult to solve. In fact, \eqref{MSMaximization} belongs to the class of binary nonconvex programming, which is generally NP-hard. For this type of problems, a pragmatic approach is to find a high-performing solution. To this end, we find that  the simple solution obtained by rounding off the solution obtained by solving \eqref{Maximization} to the nearest binary value can produce a reasonably good performance. This shall be numerically demonstrated in the next section. More advanced methods for solving \eqref{MSMaximization} are thus left for future work.
\section{Numerical Results}\label{Numerical}
\textcolor{black}{In this section, we present numerical results for the sum SE in STAR-RIS-aided systems, using both analytical techniques and Monte Carlo simulations. Specifically, our analytical results for the sum SE are derived from equations \eqref{gammaSINR}-\eqref{Den1}, while for Monte Carlo simulations, we perform 1000 independent channel realizations to evaluate the expressions in equations \eqref{gamma1}-\eqref{Int}. This is to verify  the tightness of the approximation stated in Proposition \ref{Proposition:DLSINR} and the derivations in Appendix \ref{Proposition1}. The results shown in Figs.~\ref{fig2} and \ref{fig4} clearly demonstrate a close match between the analytical results and MC simulations, and thus, confirming that Proposition \ref{Proposition:DLSINR} indeed presents a very tight approximation of the SINR.} 

The simulation setup includes a STAR-RIS with a UPA of $ N=64 $ elements assisting the communication between a uniform linear array (ULA) of $ M =64$ antennas  at the BS that serves $ K = 4 $ UEs. The $xy-$coordinates of the BS and RIS are given as $(x_B,~ y_B) = (0,~0)$ and $(x_R,~ y_R)=(50,~ 10)$, respectively, all in meter units.  In addition, users in $r$ region are located on a straight line between $(x_R-\frac{1}{2}d_0,~y_R-\frac{1}{2}d_0)$ and $(x_R+\frac{1}{2}d_o,~y_R-\frac{1}{2}d_0)$ with equal distances between each two adjacent users, and $d_0 = 20$~m in our simulations. Similarly, users in the $t$ region are located between $(x_R-\frac{1}{2}d_0,~y_R+\frac{1}{2}d_0)$ and $(x_R+\frac{1}{2}d_o,~y_R+\frac{1}{2}d_0)$. The size of each RIS element is $ d_{\mathrm{H}}\!=\!d_{\mathrm{V}}\!=\!\lambda/4 $.   Distance-based path-loss is considered in our work, such that the channel gain of a given link $j$ is $\tilde \beta_j = A d_j^{-\alpha_{j}}$, where $A$ is the area of each reflecting element at the RIS, and $\alpha_{j}$ is the path-loss exponent. Regarding $ \tilde{ \beta}_{g} $,  we assume the same values as
for $ \tilde \beta_j $.  Similar values are assumed for $ \bar{ \beta}_{k} $   but we also consider an additional penetration loss equal to $ 15 $ dB. The correlation matrices $ \bR_{\mathrm{BS}}$ and $\bR_{\mathrm{RIS}} $ are  computed according to \cite{Hoydis2013} and \cite{Bjoernson2020}, respectively.  Also, $ \sigma^2=-174+10\log_{10}B_{\mathrm{c}}  $ \textcolor{black}{ in dBm}, where $B_{\mathrm{c}}=200~\mathrm{kHz}$ is the bandwidth.

% ${ \gamma}_{\mathrm{g}} $, $\tilde{ \gamma}_{\mathrm{g}}$, $\gamma_{h_{k}} $, $\tilde{\gamma}_{h_{k}}$ 

\textcolor{black}{As a baseline scheme, we consider the  RIS, which consists of transmitting-only or reflecting-only elements, each with $ N_{t} $ and $ N_{r} $ elements, such that $ N_{t}+N_{r} =N$.} Notably, this scheme resembles  the MS protocol, where the first $ N_{t} $ elements operate in transmission mode and the $ N_{r} $ elements operate in reflection mode. \textcolor{black}{Also, we have applied an ON/OFF scheme for  channel estimation by following the idea in \cite{Mishra2019}, the direct links are estimated with all sub-surfaces turned off and the cascaded links are estimated with one  element turned on at transmission/reflection mode sequentially.}
\begin{figure}[!h]
	\begin{center}
		\includegraphics[width=0.9\linewidth]{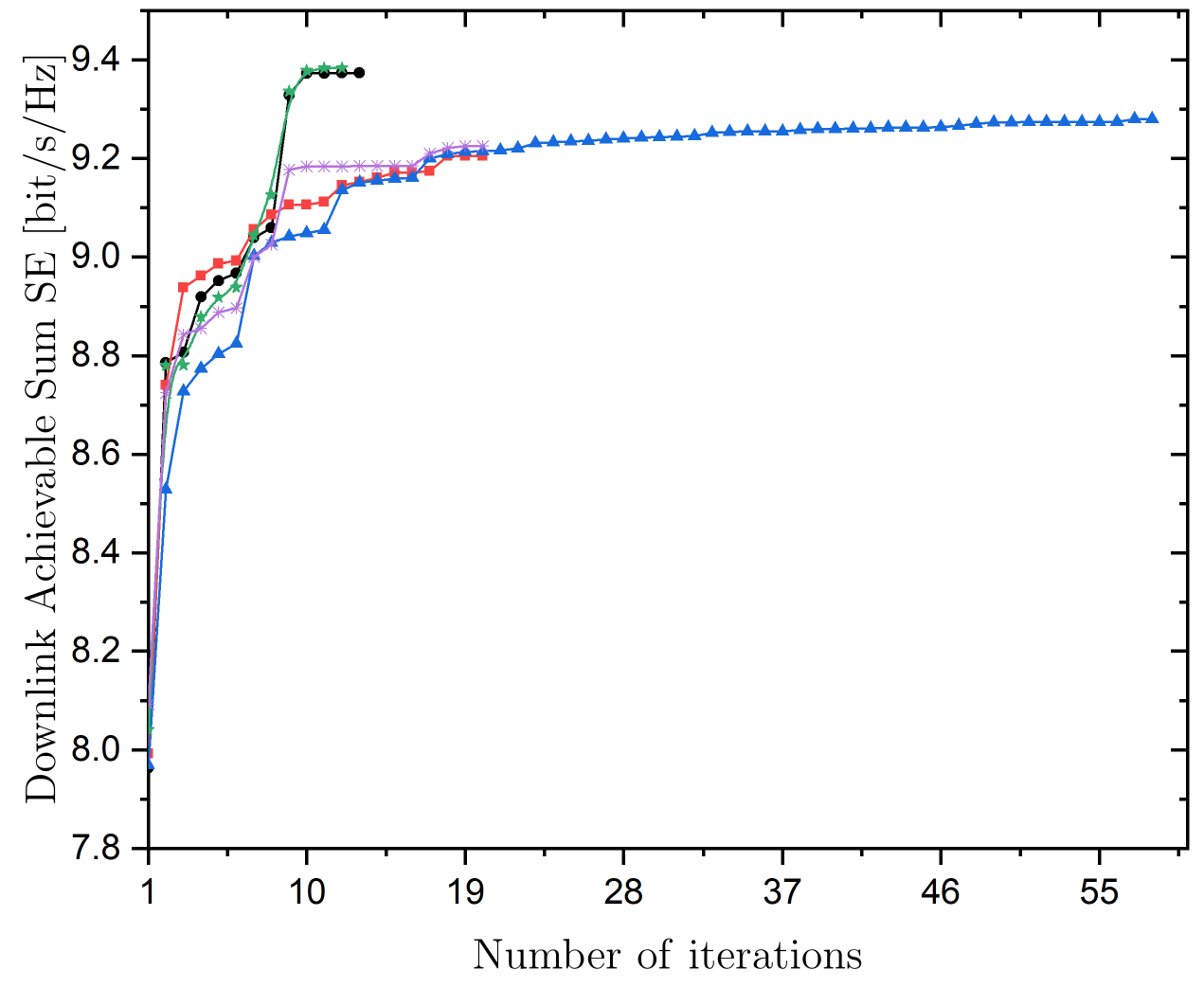}
		\caption{{Convergence of Algorithm \ref{Algoa1} for an STAR-RIS assisted MIMO system with imperfect CSI ($M=64$, $ N=64 $, $ K=4 $) \textcolor{black}{for five different initial points.} }}
		\label{fig:convergence}
	\end{center} 
\end{figure}
\begin{figure}[!h]
	\begin{center}
		\includegraphics[width=0.9\linewidth]{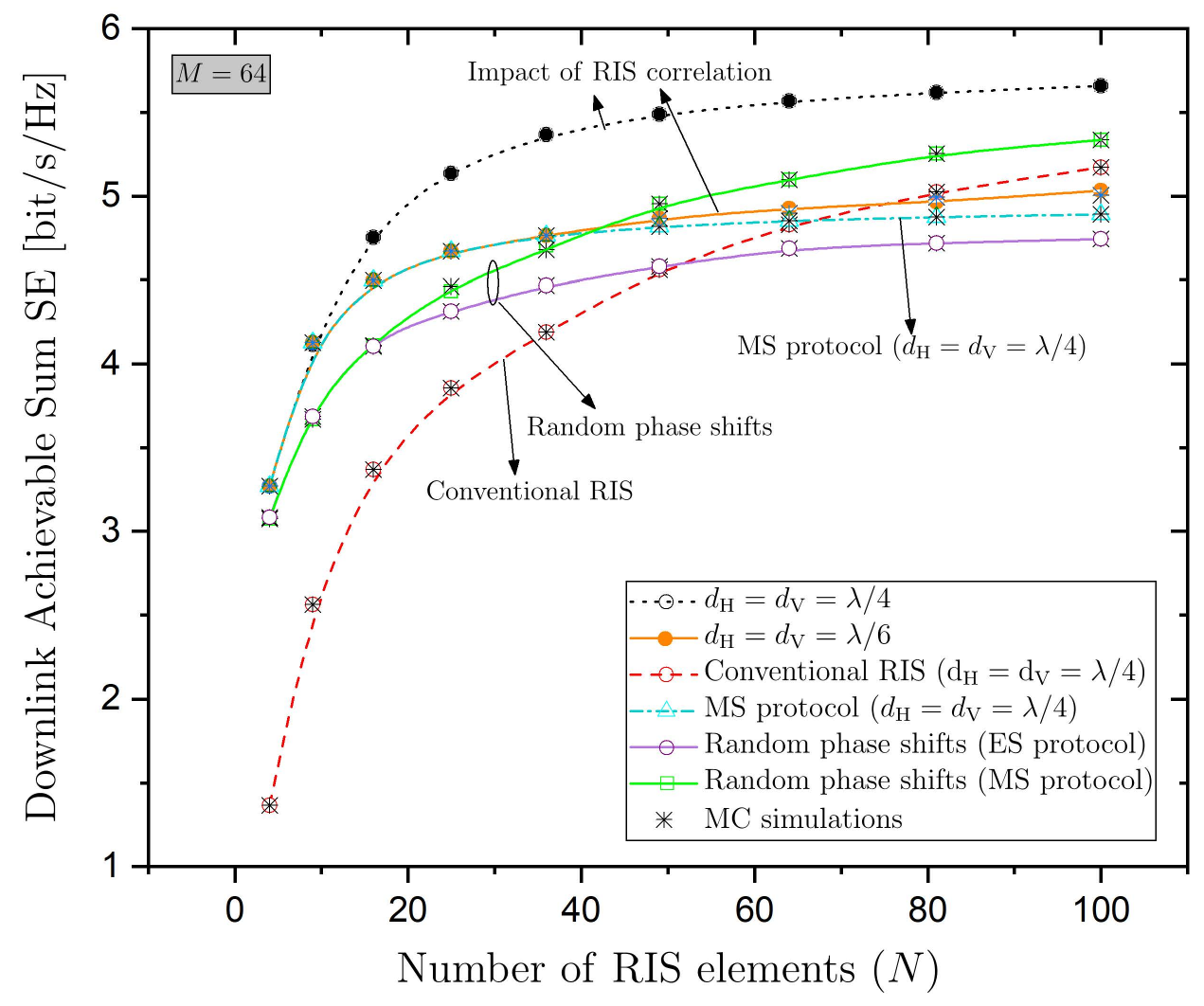}
		\caption{{Downlink achievable sum SE versus the number of RIS elements antennas $N$ of a STAR-RIS assisted MIMO system with imperfect CSI ($ N=64 $, $ K=4 $) for varying conditions (Analytical results and MC simulations). }}
		\label{fig2}
	\end{center}
\end{figure}

\begin{figure}%
	\centering
	\subfigure[]{	\includegraphics[width=0.9\linewidth]{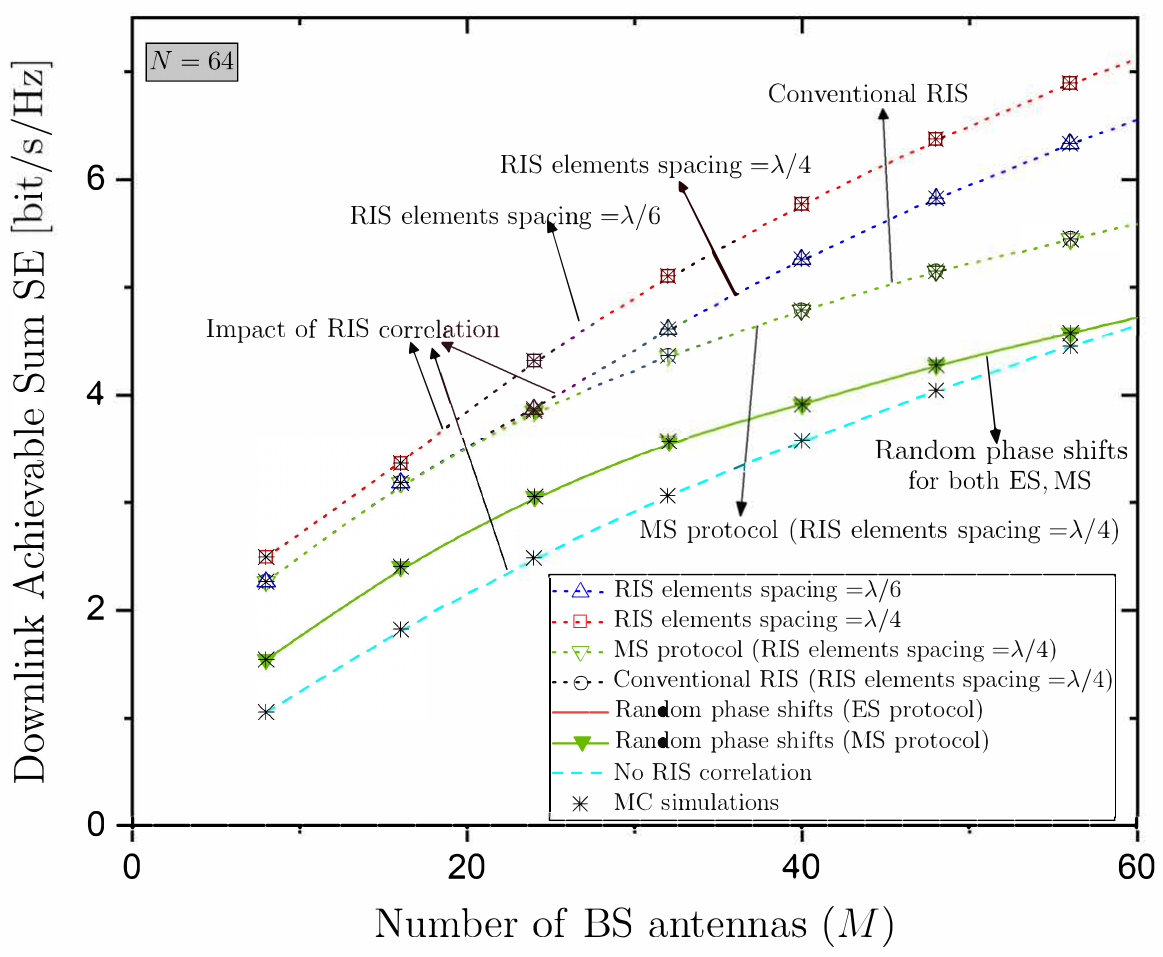}}\qquad
	\subfigure[]{	\includegraphics[width=0.9\linewidth]{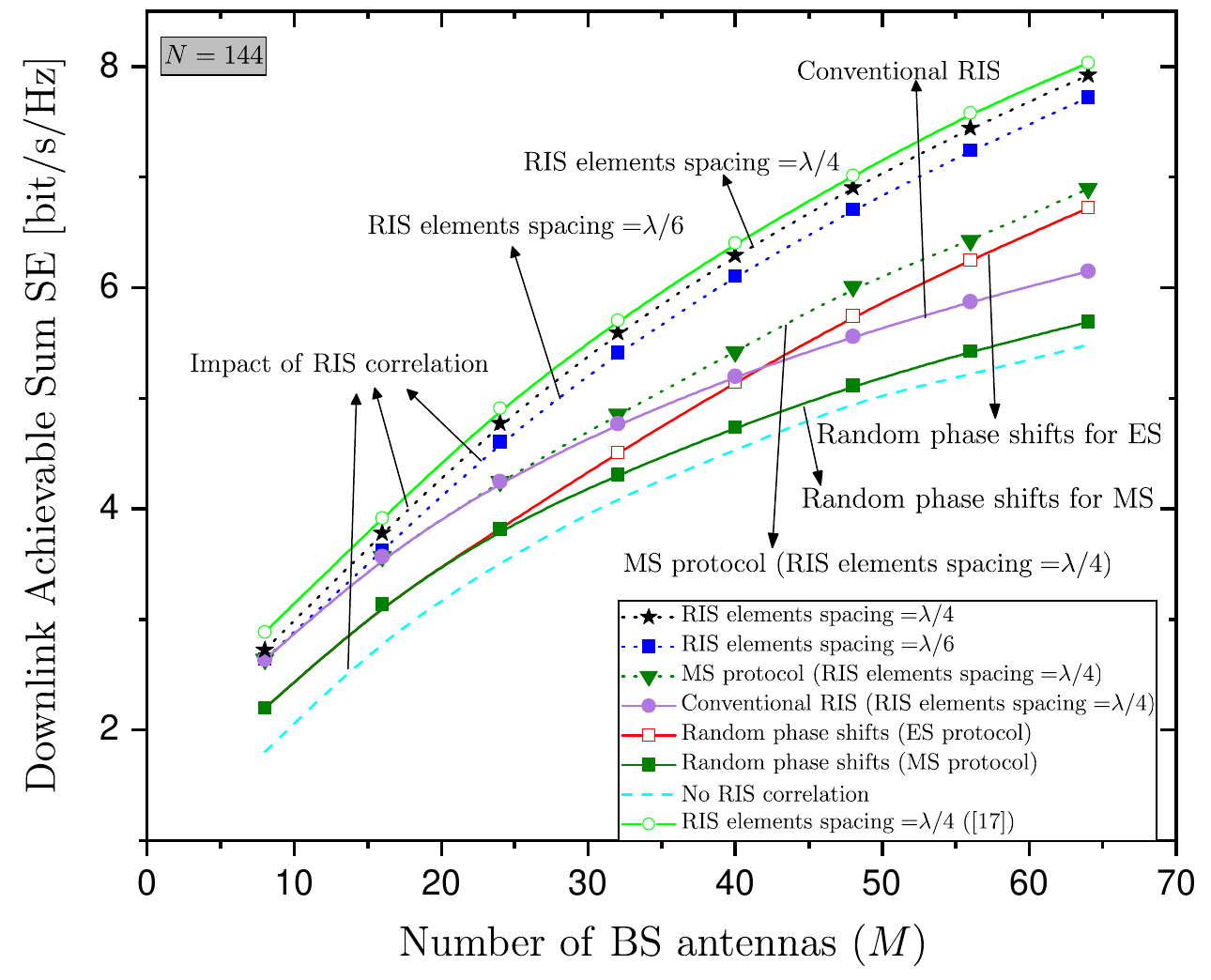}}\\
	\caption{ Downlink achievable sum SE versus the number of BS antennas $M$ of a STAR-RIS assisted MIMO system with imperfect CSI for: (a) $ N=64 $, $ K=4 $, (b) $ N=144 $, $ K=4 $ under varying conditions (Analytical results). }
	\label{fig3}
\end{figure}

\begin{figure}[!h]
	\begin{center}
		\includegraphics[width=0.9\linewidth]{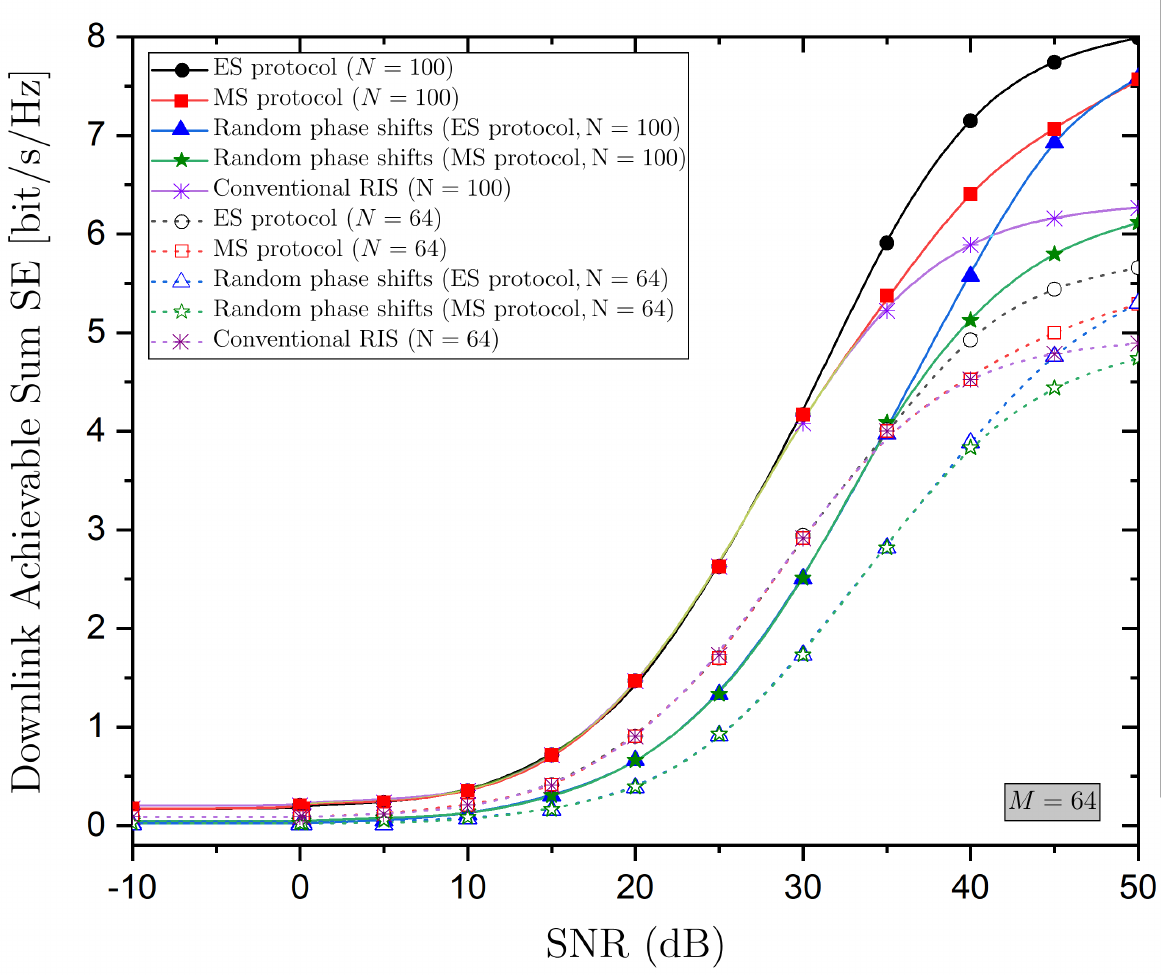}
		\caption{{Downlink achievable sum SE versus the SNR of a STAR-RIS assisted MIMO system with imperfect CSI ($M=64$, $ N=64 $, $ K=4 $) for varying conditions (Analytical results). }}
		\label{fig4}
	\end{center} 
\end{figure}
In the first numerical experiment, we demonstrate the convergence of the proposed projected gradient algorithm. Specifically, we plot the achievable sum SE against the iteration count returned by Algorithm \ref{Algoa1} from 5 different randomly generated initial points as shown in Fig.~\ref{fig:convergence}.  More specifically,  {\color{black}the   initial points for Algorithm \ref{Algoa1} are generated as follows. First, we set the amplitudes to   $[\betv^{(0)}_r]_n=[\betv^{(0)}_t]_n=\sqrt{0.5}$, for all $n\in\mathcal{N}$, i.e, equal power splitting between transmission and reception mode for all elements of the STAR-RIS. The initial values for the phase shifts are taken as  $[\thetv^{(0)}_{r}]_n=e^{j\phi_n^r}$ and  $[\thetv^{(0)}_{t}]_n=e^{j\phi_n^t}$, where $\phi_n^r$ and $\phi_n^t$ are independently drawn from the uniform distribution over $[0,2\pi]$}. We terminate Algorithm \ref{Algoa1} when the  increase of the objective between two last iterations is less than $10^{-5}$ or the number of iterations is larger than $200$. Note that the considered problem in \eqref{Maximization} is nonconvex, and thus, the proposed  projected gradient algorithm can only guarantee a stationary solution that is not necessarily optimal. As a result, Algorithm \ref{Algoa1} may converge to different points starting from different initial points, which is clearly seen in Fig.~\ref{fig:convergence}. Moreover, we can see that different initial points may lead to different convergence rates. Thus, to mitigate this performance sensitivity of Algorithm  \ref{Algoa1} on the initial points, we need to run it from different initial points and take the best convergent solutions. Through our extensive simulations, it is best to run Algorithm  \ref{Algoa1} from 5 randomly generated initial points to achieve a good trade-off between complexity and  obtained sum SE.

Fig. \ref{fig2} shows the achievable sum SE versus the number of  STAR-RIS elements $ N $ while varying the effect of spatial correlation in terms of the size of each RIS element. First, as can be seen, the downlink sum SE increases with $ N $ as expected. Next, by focusing on the impact of spatial correlation at the STAR-RIS, we show that the performance decreases as  the correlation increases. In particular, the sum SE decreases with increased correlation as the inter-element distance of the STAR-RIS decreases.   Moreover, the MS protocol   achieves a lower performance because it is a special case of the ES protocol. Especially, for a low number of RIS elements  the curves coincide, while as $N$ increases, an increasing gap appears.  
Furthermore, for the sake of comparison, we provide the performance of conventional RIS with reflection-only operation but this also appears lower performance since fewer degrees of freedom for just reflection can be exploited. We have also depicted the performance in the case of blocked direct signal. Obviously, the STAR-RIS contributes to the performance since the line corresponding the case with no direct signal is lower, which means that the performance is worse.

Figs. \ref{fig3}(a) and \ref{fig3}(b) illustrate the achievable sum SE versus the number of BS antennas $ M $ while shedding light on various effects. Obviously, the sum SE  increases with $ M $. \textcolor{black}{In particular, regarding the RIS correlation, in Fig. \ref{fig3}(a), we observe that   an increased correlation by reducing the distance among the RIS elements degrades the performance  due to reduced diversity gains among the RIS elements. In the case of no RIS correlation, represented by the dashed cyan line, the performance is quite low due to the absence of capability for phase shift optimization as mentioned in Remark 1.} Moreover, the ES protocol achieves better performance but with higher complexity compared to the MS protocol. The performance increases with more BS antennas. Also, in the case of random phase shifts, the sum SE is lower. Notably, for $N=64$ elements, the two lines corresponding to the ES and MS protocols coincide but, according to Fig.\ref{fig3}(b), which assumes $N=144$ elements, a gap between the lines appears. The gap increases with increasing $M$.  Similar to the previous figure, we have included the baseline scenario with reflecting-only capabilities having the  half elements $ N=20 $, and we witness the superiority of STAR-RIS. Moreover, a comparison of the two figures, reveals that an increase in RIS elements spacing has a greater impact on a lower number of RIS elements, i.e., $N=64$. Furthermore, Fig. \ref{fig3}(b) shows that the scenario of no RIS correlation   performs worse than the cases with random phases when $M$ is large. Also, in this figure, it is shown that the achievable rate is higher than Fig~\ref{fig3}(b). \textcolor{black}{In addition, in Fig. \ref{fig3}(b), we have added a line corresponding to channel estimation based on the ON/OFF scheme in \cite{Mishra2019}. We observe that the achievable rate is higher in this case because in the case of statistical CSI we have loss of information. In other words, we observe a trade-off between a lower overhead of the proposed approach and a higher rate in the case of estimating  individual channels.}

Fig. \ref{fig4} depicts the achievable sum SE versus the SNR under similar conditions, i.e., in the cases of $N=100$ (solid lines) and $N=64$ (dotted lines). As expected, when $N=100$, the performance is better since a higher SE is achieved. In each case, for low SNR, the ES and MS protocols coincide, while for high SNR, an increasing gap is observed. In the case $N=100$, conventional RIS and random MS protocols exhibit the same performance at low SNR, but a gap appears as the SNR increases. The behavior at low SNR is similar, however, the corresponding gaps are smaller. \textcolor{black}{The reasons for these observations can be explained as follows. At low SNR, it is more beneficial to focus on users in the reflection region as they are closer to the BS. This is confirmed by the fact that, after running the proposed algorithm, $\beta_{n}^{r}\approx 1,\forall n\in \mathcal{N}$. As a result, the performances of the ES and MS protocols, as well as the conventional RIS, are nearly the same. However, as SNR increases, the increase in the sum SE becomes minimal if we continue to focus on users in the reflection region. Thus, at high SNR, directing some power to users in the transmission region can improve the total SE. This leads to performance differences between the ES and MS protocols, as well as the conventional RIS.}

\section{Conclusion} \label{Conclusion} 
This paper presented a study of the achievable rate of STAR-RIS assisted mMIMO systems while accounting for imperfect CSI and correlated Rayleigh fading. Notably, we considered several UEs, each of which can lie on either side of the RIS, and we derived the achievable rate in closed-form. Also, we provided a low-complexity iterative optimization approach to maximizing the achievable rate, in which the amplitudes and the phase shifts of the RIS are updated simultaneously at each iteration. Furthermore, we provided useful insights into the impact of RIS correlation and showed that STAR-RIS is more beneficial compared to the traditional RIS which is reflecting only.

\begin{appendices}
	\section{Proof of Lemma~\ref{PropositionDirectChannel}}\label{lem1}
	The LMMSE estimator of $ \bh_{k} $, obtained by minimizing $ \tr\!\big(\EE\big\{\!(\hat{\bh}_{k}-{\bh}_{k})(\hat{\bh}_{k}-{\bh}_{k})^{\H}\!\big\}\!\big) $, is given  by
	\begin{align}
		\hat{\bh}_{k} =\EE\!\left\{\br_{k}\bh_{k}^{\H}\right\}\left(\EE\!\left\{\br_{k}\br_{k}^{\H}\right\}\right)^{-1}\br_{k}.\label{Cor6}
	\end{align}
	Given that the channel and the receiver noise are uncorrelated, we obtain
	\begin{align}
		\EE\left\{\br_{k}\bh_{k}^{\H}\right\}
		&=\EE\left\{\bh_{k}\bh_{k}^{\H}\right\}=\bR_{k}.\label{Cor0}
	\end{align}
	The second term in \eqref{Cor6} is written as
	\begin{align}
		\EE\left\{\br_{k}\br_{k}^{\H}\right\}&=\bR_{k} +\frac{\sigma^2}{ \tau P }\Id_{M}.\label{Cor1}
	\end{align}
	The LMMSE estimate in \eqref{estim1} is obtained by inserting \eqref{Cor0} and \eqref{Cor1} into \eqref{Cor6}, which completes the proof. We further note that  the covariance matrix of the estimated channel is
	\begin{align}
		\EE\left\{\hat{\bh}_{k}	\hat{\bh}_{k}^{\H}\right\}=\bR_{k}\bQ_{k}\bR_{k}.\label{var1}
	\end{align} 
	
	\section{Proof of Proposition~\ref{Proposition:DLSINR}}\label{Proposition1}
	
	Recalling the property  $\bx^{\H}\by = \tr(\by \bx^{\H})$ for any vectors $\bx$, $\by$, we can further rewrite $ {{S}}_k $ in \eqref{Sig} as
	\begin{align}
		{{S}}_{k}&=|\EE\{\bh_{k}^{\H}\hat{\bh}_{k}\}|^{2}=|\tr\big( \EE\{\hat{\bh}_{k} {\bh}_{k}^{\H} \} \!\big)|^{2} \\
		&=|\tr\left( \EE\left\{\bR_{k}\bQ_{k} \br_{k}{\bh}_{k}^{\H}\right\} \right)\!|^{2}\label{term1}\\
		&=|\tr\left(\bPsi_{k}\right)\!|^{2}\label{term2},
	\end{align}
	where, in \eqref{term1}, we have substituted \eqref{estim1}. The last equation is obtained after computing the expectation between $ \br_{k} $ and $ \bh_{k} $.

	Next, the first term of $ I_k $ in \eqref{Int} is written as
	\begin{align}
		&\!\!\EE\big\{ \big| {\bh}_{k}^{\H}\hat{\bh}_{k}-\EE\big\{
		{\bh}_{k}^{\H}\hat{\bh}_{k}\big\}\big|^{2}\big\}\!=\!
		\EE\big\{ \big| {\bh}_{k}^{\H}\hat{\bh}_{k}\big|^{2}\big\}\!-\!\big|\EE\big\{
		{\bh}_{k}^{\H}\hat{\bh}_{k}\big\}\big|^{2} \label{est2}\\
		&=\EE\big\{ \big| \hat{\bh}_{k}^{\H} \hat{\bh}_{k} +\tilde{\bh}_{k}^{\H}\hat{\bh}_{k}\big|^{2}\big\}-\big|\EE\big\{
		\hat{\bh}_{k}^{\H}\hat{\bh}_{k}\big\}\big|^{2}\label{est3} \\
		&{\color{black}=\EE\big\{ \big|\tilde{\bh}_{k}^{\H}\hat{\bh}_{k}\big|^{2}\big\}+2\Re\{\EE\{\hat{\bh}_{k}^{\H}\hat{\bh}_{k}\hat{\bh}_{k}^{\H}\tilde{\bh}_{k}\}\}\label{est31}} 
		% 		&=\EE\big\{|\tilde{\bh}_{k}^{\H}\hat{\bh}_{k}|^{2}\big\} \label{est5}\\
		% &=\tr\!\left( \bR_{k}\bPsi_{k}\right)-\tr\left( \bPsi_{k}^{2}\right),\label{est4}
	\end{align}
	where in~\eqref{est3}, we have used (9) in the main text. \textcolor{black}{	In~\eqref{est31}, to simplify the second term, we resort to the well-known  channel hardening property in massive MIMO which intuitively states that channels behave as deterministic. The same property is also applied to the estimated channels, which means $\hat{\bh}_{k}^{\H}\hat{\bh}_{k} \approx\EE\{\hat{\bh}_{k}^{\H}\hat{\bh}_{k}\} $ with high accuracy \cite{Bjoernson2017}. Using this property, we have
		\begin{align}
			\EE\{\hat{\bh}_{k}^{\H}\hat{\bh}_{k}\hat{\bh}_{k}^{\H}\tilde{\bh}_{k}\}\approx\EE\{\hat{\bh}_{k}^{\H}\hat{\bh}_{k}\}\EE\{\hat{\bh}_{k}^{\H}\tilde{\bh}_{k}\}=0
		\end{align}
		and thus
		\begin{align}
			&\!\!\EE\big\{ \big| {\bh}_{k}^{\H}\hat{\bh}_{k}-\EE\big\{
			{\bh}_{k}^{\H}\hat{\bh}_{k}\big\}\big|^{2}\big\}\!\approx \!\EE\big\{|\tilde{\bh}_{k}^{\H}\hat{\bh}_{k}|^{2}\big\} \label{est5}\\
			&=\tr\!\left( \bR_{k}\bPsi_{k}\right)-\tr\left( \bPsi_{k}^{2}\right).\label{est4}
		\end{align}
		% where we have approximated $ \hat{\bh}_{k}^{\H}\hat{\bh}_{k} $ by a deterministic term $ q $ according to the law of large numbers holding in mMIMO systems.
		%	 and that $\EE\left\{ |X+Y|^{2}\right\} =\EE\left\{ |X|^{2}\right\} +\EE\left\{ |Y^{2}|\right\}$, which is valid for any two uncorrelated random variables while one of them has zero mean value.
	}We note that \eqref{est4} holds  because $  \EE\{|\tilde{\bh}_{k}^{\H}\hat{\bh}_{k}|^2\}=  \tr(\EE\{\tilde{\bh}_{k}\tilde{\bh}_{k}^{\H}\hat{\bh}_{k}\hat{\bh}_k^{\H}\}){\color{black}\approx} \tr(\EE\{\tilde{\bh}_k\tilde{\bh}_k^{\H}\}\EE\{\hat{\bh}_k\hat{\bh}_k^{\H}\})
	= \tr((\bR_k-\bPsi_k)\bPsi_k)
	$, {\color{black}  where we have applied the approximations $\hat{\bh}_{k}^{\H}\hat{\bh}_{k} \approx\EE\{\hat{\bh}_{k}^{\H}\hat{\bh}_{k}\} $ and $\tilde{\bh}_{k}^{\H}\tilde{\bh}_{k} \approx\EE\{\tilde{\bh}_{k}^{\H}\tilde{\bh}_{k}\} $, which is due to the channel hardening property in massive MIMO as explained above}.

	\textcolor{black}{	For the second term of $ I_k $ in \eqref{Int} it is easy to check that
		\begin{align}
			&\EE\big\{ \big| {\bh}_{k}^{\H}\hat{\bh}_{i}\big|^{2}\big\}=\EE\big\{ \big| (\hat{\bh}_{k}^{\H}+\tilde{\bh}_{k})\hat{\bh}_{i}\big|^{2}\big\}\\
			&=\EE\big\{ \big| \hat{\bh}_{k}^{\H}\hat{\bh}_{i}\big|^{2}\big\}+\EE\big\{ \big|\tilde{\bh}_{k}\hat{\bh}_{i}\big|^{2}\big\}+2\Re\{\EE\{\hat{\bh}_{k}^{\H}\hat{\bh}_{i}\hat{\bh}_{i}^{\H}\tilde{\bh}_{k}\}\}\label{52}\\
			&=		\tr\!\left(\bR_{k}\bPsi_{i} \right).\label{54}
		\end{align}
		We note that \eqref{54} is true because the third term in \eqref{52} is zero as can be shown below.
		Specifically, this term can be written as
		\begin{align}
		\!\!	2\Re\{\EE\{\hat{\bh}_{k}^{\H}\hat{\bh}_{i}\hat{\bh}_{i}^{\H}\tilde{\bh}_{k}\}\}&\!=\!2\Re\left\{\EE\{\tr\!\left(\!\left(\hat{\bh}_{i}\hat{\bh}_{i}^{\H}\right)\!\!\left(\tilde{\bh}_{k}\hat{\bh}_{k}^{\H}\right)\!\right)\!\right\}\\
			&\!=\!2\Re\left\{\tr\!\left(\EE\{\hat{\bh}_{i}\hat{\bh}_{i}^{\H}\}\EE\{\tilde{\bh}_{k}\hat{\bh}_{k}^{\H}\}\right)\!\right\}\label{521}\\
			&=0\label{522}
		\end{align}
		where, in \eqref{521}, we have accounted for the independence between $ \hat{\bh}_{k} $ and $ \hat{\bh}_{i} $, and, in \eqref{522}, we have considered that $ \tilde{\bh}_{k} $ and $ \hat{\bh}_{k} $ are uncorrelated.}
	
	The normalization parameter is written as
	\begin{align}
		\!\!\!	\lambda=\frac{1}{\sum_{i=1}^{K}\!\EE\{\mathbf{f}_{i}^{\H}\mathbf{f}_{i}\}}=\frac{1}{\sum_{i=1}^{K}\!\mathbb{E}\{\hat{\mathbf{h}}_{i}^{\H}\hat{\mathbf{h}}_{i}\}}=\frac{1}{\sum_{i=1}^{K}\!\tr(\boldsymbol{\Psi}_{i})}.\label{normalization}
	\end{align}
	{\color{black}	 Combining \eqref{est4},  \eqref{54}, and \eqref{normalization}, we can approximate $I_k$ as  $\tilde{I}_k$ in \eqref{Den1}, and thus complete the proof.}
	% The proof concludes by substituting \eqref{term2},  \eqref{est4},   into \eqref{Sig} and \eqref{Int}.
	\section{Proof of Lemma~\ref{LemmaGradients}}\label{lem2}
	Let us first derive $\nabla_{\boldsymbol{\theta^{t}}}f(\thetv,\betv) $ the complex gradient of the achievable sum SE with respect to $ \boldsymbol{\theta}^{t\ast}$. From \eqref{LowerBound}, it is easy to see that 	
	\begin{equation}
		\nabla_{\thetv^{t}}f(\thetv,\betv)=c\sum_{k=1}^{K}\frac{{\color{black}\tilde{I}_{k}}\nabla_{\boldsymbol{\theta}^{t}}S_{k}-S_{k}\nabla_{\boldsymbol{\theta}^{t}}{\color{black}\tilde{I}_{k}}}{(1+\gamma_{k}){\color{black}\tilde{I}_{k}}^{2}},
	\end{equation}
	where $c=\frac{\tau_{c}-\tau}{\tau_{c}\log_{2}(e)}$. To compute $\nabla_{\thetv^{t}}S_k$ for a given user $k$, we immediately note that $\nabla_{\thetv^{t}}S_k = 0$ if $w_k=r$, i.e., if UE $k$ is in the reflection region. This is obvious from \eqref{Num1}, \eqref{Psiexpress} and \eqref{cov1}. Thus, we only need to find $\nabla_{\thetv^{t}}S_k$ when $w_k=t$. In such a case, we can explicitly write $\bR_{k}$
	\begin{align}
		\bR_{k}&=\bar{ \beta}_{k}\bR_{\mathrm{BS}}+\hat{\beta}_{k}\tr(\bR_{\mathrm{RIS}} \bPhi_{t} \bR_{\mathrm{RIS}}  \bPhi_{t}^{\H})\bR_{\mathrm{BS}}\nn\\ 
		&=\bar{ \beta}_{k}\bR_{\mathrm{BS}}+\hat{\beta}_{k}\tr(\bA_{t}\bPhi_{t}^{\H}  )\bR_{\mathrm{BS}}\label{cov1t}.
	\end{align}
	where $\hat{\beta}_{k}=\tilde{\beta}_{g}\tilde{\beta}_{k}$ and $\mathbf{A}_{t}=\mathbf{R}_{\mathrm{RIS}}\bPhi_{t}\mathbf{R}_{\mathrm{RIS}}$. When  $w_k=r$, we define $\mathbf{A}_{r}=\mathbf{R}_{\mathrm{RIS}}\bPhi_{r}\mathbf{R}_{\mathrm{RIS}}$.
	% We rewrite the achievable sum SE as
	% 	\begin{align}
		% 		\mathrm{SE} &  =c\Bigl(\sum\nolimits _{k=1}^{K}\log\bigl(S_{k}+I_{k}\bigr)-\log I_{k}\Bigr),
		% 	\end{align}
	% . 
	% 	Thus, 	the  complex gradient of  $ f(\thetv,\betv) $ with respect to $ \thetv^{*} $
	% 	\begin{align}
		% 		&\nabla_{\boldsymbol{\theta^{\ast}}}f(\thetv,\betv)  \triangleq\frac{\partial}{\partial\boldsymbol{\theta}^{\ast}}f(\thetv,\betv)\nn\\
		% 		&=c\bigl(\frac{1}{S_{k}+I_{k}}\bigl(\frac{\partial}{\partial\boldsymbol{\theta}^{\ast}}S_{k}+\frac{\partial}{\partial\boldsymbol{\theta}^{\ast}}I_{k}\bigr)-\frac{1}{I_{k}}\frac{\partial}{\partial\boldsymbol{\theta}^{\ast}}I_{k}\bigr)\nonumber \\
		% 		& =c\Bigl(\frac{1}{S_{k}+I_{k}}\frac{\partial}{\partial\boldsymbol{\theta}^{\ast}}S_{k}-\frac{S_{k}}{I_{k}(S_{k}+I_{k})}\frac{\partial}{\partial\boldsymbol{\theta}^{\ast}}I_{k}\Bigr),
		% 	\end{align}
	% where $ f(\cdot) $ is the sum  SE given by \eqref{LowerBound}. 
	
	To calculate $\nabla_{\thetv^{t}}S_k$, we follow steps detailed in \cite[Chap. 3]{hjorungnes:2011}. First, let us denote $d(\cdot)$  the complex differential of the function in the argument.	
	% To find $\frac{\partial}{\partial\boldsymbol{\theta}^{\ast}}S_{k}$,
	% 	we note that
	Then, it holds that	\begin{align}
		d(S_{k})&=d\bigl(\tr(\boldsymbol{\Psi}_{k})^{2}\bigr)\nn\\
		&=2\tr(\boldsymbol{\Psi}_{k})d\tr(\boldsymbol{\Psi}_{k})\nn\\
		&=2\tr(\boldsymbol{\Psi}_{k})\tr(d\boldsymbol{\Psi}_{k}).\label{eq:dSk}
	\end{align}
	% 	where $dq(\cdot)$ denotes the differential of the function $q(\cdot)$.
	Next, we apply \cite[Eq. (3.35)]{hjorungnes:2011}, which gives
	\begin{align}
		&d(\boldsymbol{\Psi}_{k})=d(\mathbf{R}_{k}\mathbf{Q}_{k}\mathbf{R}_{k})\nn\\
		&=d(\mathbf{R}_{k})\mathbf{Q}_{k}\mathbf{R}_{k}+\mathbf{R}_{k}d(\mathbf{Q}_{k})\mathbf{R}_{k}+\mathbf{R}_{k}\mathbf{Q}_{k}d(\mathbf{R}_{k}).\label{eq:dPsik}
	\end{align}
	The differentials $d(\mathbf{R}_{k})$ and $d(\mathbf{Q}_{k})$ are derived as follows. First, from \eqref{cov1t} it is easy to check that 
	\begin{align}
		d(\mathbf{R}_{k})\label{dRk:general}  %&=\hat{\beta}_{k}\mathbf{R}_{\mathrm{BS}}d\bigl(\tr\bigl(\mathbf{R}_{\mathrm{RIS}}\bPhi_{t}\mathbf{R}_{\mathrm{RIS}}\bPhi_{t}\herm\bigr)\bigr)\nn\\
		% 		& =\hat{\beta}_{k}\mathbf{R}_{\mathrm{BS}}\big(\tr\bigl(\mathbf{R}_{\mathrm{RIS}}d(\bPhi_{t})\mathbf{R}_{\mathrm{RIS}}\bPhi_{t}\herm\bigr)+\tr\bigl(\mathbf{R}_{\mathrm{RIS}}\bPhi_{t}\mathbf{R}_{\mathrm{RIS}}d\bigl(\bPhi_{t}\herm\bigr)\bigr)\big)\nn\\
		% 		& =\hat{\beta}_{k}\mathbf{R}_{\mathrm{BS}}\big(\tr\bigl(\mathbf{R}_{\mathrm{RIS}}\bPhi_{t}\herm\mathbf{R}_{\mathrm{RIS}}d(\bPhi_{t})\bigr)+\tr\bigl(\mathbf{R}_{\mathrm{RIS}}\bPhi_{t}\mathbf{R}_{\mathrm{RIS}}d\bigl(\bPhi_{t}\herm\bigr)\bigr)\big)\nn\\
		% 		
		& =\hat{\beta}_{k}\mathbf{R}_{\mathrm{BS}}\tr\bigl(\mathbf{A}_{t}\herm d(\bPhi_{t})+\mathbf{A}_{t}d\bigl(\bPhi_{t}\herm\bigr)\bigr),
	\end{align}

	% 	Moreover, let $\boldsymbol{\mathbf{\beta}}=[\sqrt{\beta_{1}},\sqrt{\beta_{2}},\ldots,\sqrt{\beta_{N}}]\trans$.
	Since $\bPhi_{t}$ is diagonal,  we can further write  $d(\mathbf{R}_{k})$ as
	\begin{align}
		d(\mathbf{R}_{k})&=\hat{\beta}_{k}\mathbf{R}_{\mathrm{BS}}\bigl(\bigl(\diag\bigl(\mathbf{A}_{t}\herm\diag(\boldsymbol{{\beta}}^{t})\bigr)\bigr)\trans d(\boldsymbol{\theta}^{t})\nn\\
		&+\bigl(\diag\bigl(\mathbf{A}_{t}\diag(\boldsymbol{{\beta}}^{t})\bigr)\bigr)\trans d(\boldsymbol{\theta}^{t\ast})\bigr).\label{eq:dRk}
	\end{align}
	% 	where $(\mathbf{X})_{d}$ denotes the vector of the diagonal elements
	% 	of $\mathbf{X}$ and $\diag(\mathbf{x})$ denotes the diagonal matrix
	% 	with the elements of vector $\mathbf{x}$ on the main diagonal.
	Next, we use \cite[eqn. (3.40)]{hjorungnes:2011} to obtain
	\begin{align}
		&d(\mathbf{Q}_{k})  =d\bigl(\mathbf{R}_{k}+\frac{\sigma}{\tau P}\mathbf{I}_{M}\bigr)^{-1}\nn\\
		&=-\bigl(\mathbf{R}_{k}+\frac{\sigma^{2}}{\tau P}\mathbf{I}_{M}\bigr)^{-1}d\bigl(\mathbf{R}_{k}+\frac{\sigma^{2}}{\tau P}\mathbf{I}_{M}\bigr)\bigl(\mathbf{R}_{k}+\frac{\sigma}{\tau P}\mathbf{I}_{M}\bigr)^{-1}\nonumber \\
		& =-\mathbf{Q}_{k}d(\mathbf{R}_{k})\mathbf{Q}_{k}.\label{eq:dQk}
	\end{align}
	Combining (\ref{eq:dPsik}) and (\ref{eq:dQk})
	yields
	\begin{equation}
		d(\boldsymbol{\Psi}_{k})=d(\mathbf{R}_{k})\mathbf{Q}_{k}\mathbf{R}_{k}\!-\!\mathbf{R}_{k}\mathbf{Q}_{k}d(\mathbf{R}_{k})\mathbf{Q}_{k}\mathbf{R}_{k}+\mathbf{R}_{k}\mathbf{Q}_{k}d\mathbf{R}_{k}.\label{eq:dPsik-1}
	\end{equation}
	Thus, by inserting \eqref{eq:dPsik-1} and (\ref{eq:dRk}) into \eqref{eq:dSk}, we obtain
	\begin{subequations}\label{eq:dSk:final}	\begin{align}
			&d(S_{k})  =2\tr(\boldsymbol{\Psi}_{k})\bigl(\tr(\mathbf{Q}_{k}\mathbf{R}_{k}d(\mathbf{R}_{k}))+\tr(\mathbf{R}_{k}\mathbf{Q}_{k}d(\mathbf{R}_{k}))\nn\\
			&-\tr\bigl(\mathbf{Q}_{k}\mathbf{R}_{k}^{2}\mathbf{Q}_{k}d(\mathbf{R}_{k})\bigr)\bigr)\\
			& =\nu_{k}\bigl(\diag\bigl(\bigl(\mathbf{A}_{t}\herm\diag(\boldsymbol{\mathbf{\beta}}^{t})\bigr)\bigr)\trans d\boldsymbol{\theta}^t+\bigl(\diag\bigl(\mathbf{A}_{t}\diag(\boldsymbol{\mathbf{\beta}}^{t})\bigr)\bigr)\trans d\boldsymbol{\theta}^{t\ast}\bigr),
		\end{align}
	\end{subequations}	where 
	\begin{equation}
		\nu_{k}\!=\!2\hat{\beta}_{k}\tr(\boldsymbol{\Psi}_{k})\!\tr\bigl(\!\bigl(\mathbf{Q}_{k}\mathbf{R}_{k}+\mathbf{R}_{k}\mathbf{Q}_{k}-\mathbf{Q}_{k}\mathbf{R}_{k}^{2}\mathbf{Q}_{k}\bigr)\mathbf{R}_{\mathrm{BS}}\bigr).
	\end{equation}
	From (\ref{eq:dSk:final}), we can conclude that 
	\begin{align}
		\nabla_{\thetv^{t}}S_k&=\frac{\partial}{\partial{\thetv^{t\ast}}}S_{k}\nn\\
		&=\nu_{k}\diag\bigl(\mathbf{A}_{t}\diag(\boldsymbol{{\beta}}^{t})\bigr)
	\end{align}
	for $w_k=t$, which indeed proves \eqref{derivtheta_t}. Following the same procedure we can easily prove \eqref{derivtheta_r}. The details are skipped for the sake of brevity.
	
	Now we turn our attention to $\nabla_{\thetv^{t}}{\color{black}\tilde{I}_{k}}$.	To this end, from \eqref{Den1},
	it is straightforward to check that 
	\begin{align}
		d({\color{black}\tilde{I}_{k}})  &=\sum\nolimits _{i=1}^{K}\tr(d(\mathbf{R}_{k})\boldsymbol{\Psi}_{i})+\sum\nolimits _{i=1}^{K}\tr(\mathbf{R}_{k}d(\boldsymbol{\Psi}_{i}))\nn\\
		&\quad-2\tr\bigl(\boldsymbol{\Psi}_{k}d(\boldsymbol{\Psi}_{k})\bigr)+\frac{K\sigma^{2}}{\rho}\sum\nolimits _{i=1}^{K}\tr(d(\boldsymbol{\Psi}_{i})).
		% 		& =\tr(\boldsymbol{\Psi}d\mathbf{R}_{k})+\sum\nolimits _{i=1}^{K}\tr(\bar{\mathbf{R}}_{k}d\boldsymbol{\Psi}_{i})-2\tr\bigl(\boldsymbol{\Psi}_{k}d\boldsymbol{\Psi}_{k}\bigr).
	\end{align}
	Note that $d(\boldsymbol{\Psi}_{i})=0$ if $w_{i}\neq t$ since $\boldsymbol{\Psi}_{i}$ is independent of $\thetv^{t}$ in this case. Thus, the
	above equation is reduced to 
	\begin{equation}
		d({\color{black}\tilde{I}_{k}})=\tr(\boldsymbol{\Psi}d(\mathbf{R}_{k}))-2\tr\bigl(\boldsymbol{\Psi}_{k}d(\boldsymbol{\Psi}_{k})\bigr)+\sum\nolimits _{i\in\mathcal{K}_{t}}\tr(\bar{\mathbf{R}}_{k}d(\boldsymbol{\Psi}_{i})),	\label{ik1}
	\end{equation}	where $\boldsymbol{\Psi}=\sum\nolimits _{i=1}^{K}\boldsymbol{\Psi}_{i}$
	and $\bar{\mathbf{R}}_{k}=\mathbf{R}_{k}+\frac{K\sigma^{2}}{\rho}\mathbf{I}_{M}$.
	Using \eqref{eq:dPsik-1} into \eqref{ik1} gives
	\begin{align}
		&d({\color{black}\tilde{I}_{k}}) =\tr(\boldsymbol{\Psi}d\mathbf{R}_{k})\nn\\
		&\!+\!\!\sum\nolimits _{i\in \mathcal{K}_t}\!\!\!\tr(\bar{\mathbf{R}}_{k}\bigl(d(\mathbf{R}_{i})\mathbf{Q}_{i}\mathbf{R}_{i}-\mathbf{R}_{i}\mathbf{Q}_{i}d(\mathbf{R}_{i})\mathbf{Q}_{i}\mathbf{R}_{i}+\mathbf{R}_{i}\mathbf{Q}_{i}d(\mathbf{R}_{i})\!\bigr)\!)\nn\\
		&-2\tr\bigl(\boldsymbol{\Psi}_{k}\bigl(d(\mathbf{R}_{k})\mathbf{Q}_{k}\mathbf{R}_{k}-\mathbf{R}_{k}\mathbf{Q}_{k}d(\mathbf{R}_{k})\mathbf{Q}_{k}\mathbf{R}_{k}+\mathbf{R}_{k}\mathbf{Q}_{k}d(\mathbf{R}_{k})\!\bigr)\!\bigr)\nn\\
		& =\tr(\check{\boldsymbol{\Psi}}_{k}d(\mathbf{R}_{k}))+\sum\nolimits _{i\in \mathcal{K}_t}\tr\bigl(\tilde{\mathbf{R}}_{ki}d(\mathbf{R}_{i})\bigr),
	\end{align}
	where 
	\begin{equation}
		\check{\boldsymbol{\Psi}}_{k}=\boldsymbol{\Psi}-2\bigl(\mathbf{Q}_{k}\mathbf{R}_{k}\boldsymbol{\Psi}_{k}+\boldsymbol{\Psi}_{k}\mathbf{R}_{k}\mathbf{Q}_{k}-\mathbf{Q}_{k}\mathbf{R}_{k}\boldsymbol{\Psi}_{k}\mathbf{R}_{k}\mathbf{Q}_{k}\bigr)
	\end{equation}
	and 
	\begin{equation}
		\tilde{\mathbf{R}}_{ki}=\mathbf{Q}_{i}\mathbf{R}_{i}\bar{\mathbf{R}}_{k}-\mathbf{Q}_{i}\mathbf{R}_{i}\bar{\mathbf{R}}_{k}\mathbf{R}_{i}\mathbf{Q}_{i}+\bar{\mathbf{R}}_{k}\mathbf{R}_{i}\mathbf{Q}_{i},i \in \mathcal{K}_t.
	\end{equation}
	Again, we note that $d(\bR_{k})=0$ if $w_{k}\neq t$. Thus, 	by using (\ref{eq:dRk}),  we can write $\nabla_{\thetv^{t}}{\color{black}\tilde{I}_{k}}$ as
	\begin{align}
		\nabla_{\thetv^{t}}{\color{black}\tilde{I}_{k}}=\frac{\partial}{\partial\boldsymbol{\theta}^{t\ast}}{\color{black}\tilde{I}_{k}} & =\diag\bigl(\tilde{\mathbf{A}}_{kt}\diag(\boldsymbol{\mathbf{\beta}}^{t})\bigr),
	\end{align}
	where $\bar{\nu}_{k}=\hat{\beta}_{k}\tr\bigl(\check{\boldsymbol{\Psi}}_{k}\mathbf{R}_{\mathrm{BS}}\bigr)$,
	$\tilde{\nu}_{ki}=\hat{\beta}_{k}\tr\bigl(\tilde{\mathbf{R}}_{ki}\mathbf{R}_{\mathrm{BS}}\bigr)$,
	and 
	\begin{equation}
		\tilde{\mathbf{A}}_{kt}=\begin{cases}
			\bar{\nu}_{k}\mathbf{A}_{t}+\sum\nolimits _{i\in\mathcal{K}_{t}}^{K}\tilde{\nu}_{ki}\mathbf{A}_{t} & w_{k}=t\\
			\sum\nolimits _{i\in\mathcal{K}_{t}}\tilde{\nu}_{ki}\mathbf{A}_{t} & w_{k}\neq t,
		\end{cases}
	\end{equation}
	which is in fact the special case of \eqref{A_tilde_general} when $u=t$, meaning that \eqref{derivtheta_t_Ik} has been proved. Following the same steps we can prove \eqref{derivtheta_r_Ik}, but again, we  skip the details for the sake of brevity.
	
	The expression for $\nabla_{\boldsymbol{\beta}^{t}}S_{k}$ is derived
	as follows. First we only need to consider $\nabla_{\boldsymbol{\beta}^{t}}S_{k}$
	when $w_{k}=t$. For this case, from \eqref{dRk:general}, we can write 
	\begin{subequations}\label{dRk_beta_t}
		\begin{align}
			d(\mathbf{R}_{k}) & =\hat{\beta}_{k}\mathbf{R}_{\mathrm{BS}}\tr\bigl(\mathbf{A}_{t}\herm d(\bPhi_{t})+\mathbf{A}_{t}d\bigl(\bPhi_{t}\herm\bigr)\bigr)\\
			& =\hat{\beta}_{k}\mathbf{R}_{\mathrm{BS}}\bigl(\diag\bigl(\mathbf{A}_{t}\herm\diag(\btheta^{t})\bigr)^{\T}d(\boldsymbol{\beta}^{t})\nn\\
			&+\diag\bigl(\mathbf{A}_{t}\diag(\btheta^{t\ast})\bigr)^{\T}d(\boldsymbol{\beta}^{t})\bigr)\\
			& =2\hat{\beta}_{k}\mathbf{R}_{\mathrm{BS}}\Re\bigl\{\diag\bigl(\mathbf{A}_{t}\herm\diag(\btheta^{t})\bigr\}^{\T} d(\boldsymbol{\beta}^{t}).
		\end{align}
	\end{subequations}
	Now, using \eqref{dRk_beta_t} in \eqref{eq:dSk:final} yields
	\begin{equation}
		\nabla_{\boldsymbol{\beta}^{t}}S_{k} = 2 \nu_k\Re\bigl\{\diag\bigl(\mathbf{A}_{t}\herm\diag(\btheta^{t})\bigr)\bigr\}.
	\end{equation}
	Similarly, we can write  $\nabla_{\boldsymbol{\beta}^{r}}S_{k}$ as 
	\begin{equation}
		\nabla_{\boldsymbol{\beta}^{r}}S_{k} = 2 \nu_k\Re\bigl\{\diag\bigl(\mathbf{A}_{r}\herm\diag(\btheta^{r})\bigr)\bigr\}.
	\end{equation}
	For $\nabla_{\boldsymbol{\beta}^{t}}{\color{black}\tilde{I}_{k}}$ and $\nabla_{\boldsymbol{\beta}^{r}}{\color{black}\tilde{I}_{k}}$, we can follow the same steps above, which gives
	\begin{align}
		\nabla_{\boldsymbol{\beta}^{t}}{\color{black}\tilde{I}_{k}}&=2\Re\bigl\{\diag\bigl(\tilde{\mathbf{A}}_{kt}\herm\diag(\boldsymbol{\btheta}^{t})\bigr)\bigr\}.\\
		\nabla_{\boldsymbol{\beta}^{r}}{\color{black}\tilde{I}_{k}}&=2\Re\bigl\{\diag\bigl(\tilde{\mathbf{A}}_{kr}\herm\diag(\boldsymbol{\btheta}^{r})\bigr)\bigr\}.
	\end{align}

\end{appendices}
\bibliographystyle{IEEEtran}

\bibliography{IEEEabrv,bibl1}

\end{document}